\documentclass{ieeeaccess}
\usepackage{amsmath,amssymb,amsfonts}
\usepackage{graphicx}
\usepackage{textcomp}
\usepackage{pdflscape}
\usepackage{multicol}
\usepackage{soul}
\usepackage{upquote}
\usepackage{hyperref}
\usepackage{upgreek}
\usepackage{enumitem}
\usepackage[utf8]{inputenc}
\usepackage{algorithm}
\usepackage{algpseudocode}
\usepackage{float}
\usepackage{booktabs,siunitx}
\usepackage{threeparttable}
\usepackage{dashrule}
\usepackage[short, nocomma]{optidef}
\usepackage{multirow}
\usepackage{comment}
\usepackage{tabularx}
\usepackage{pifont}
\newcommand{\cmark}{\ding{51}}
\newcommand{\xmark}{\ding{55}}
\usepackage[backend=biber,style=ieee]{biblatex}
\renewcommand{\algorithmicrequire}{\textbf{Input:}}
\renewcommand{\algorithmicensure}{\textbf{Output:}}
\algnewcommand{\LineComment}[1]{\(\triangleright\) \textcolor{blue}{\textit{#1}}}

\algdef{SE}[DOWHILE]{Do}{doWhile}{\algorithmicdo}[1]{\algorithmicwhile\ #1}
\algnewcommand\algorithmicforeach{\textbf{for each}}
\algdef{S}[FOR]{ForEach}[1]{\algorithmicforeach\ #1\ \algorithmicdo}
\def\BibTeX{{\rm B\kern-.05em{\sc i\kern-.025em b}\kern-.08em    T\kern-.1667em\lower.7ex\hbox{E}\kern-.125emX}}
\usepackage{amsthm}
\newtheorem{prop}{Proposition}

\theoremstyle{definition}

\AfterEndEnvironment{definition}{\noindent\ignorespaces}
\AfterEndEnvironment{obv}{\noindent\ignorespaces}

\begin{document}
\title{Delay-Aware Multi-Stage Edge Server Upgrade with Budget Constraint}
\author{
\uppercase{Endar~S.~Wihidayat}\authorrefmark{1}\authorrefmark{2}, 
\IEEEmembership{Member, IEEE},
\uppercase{Sieteng Soh}\authorrefmark{2}, \uppercase{Kwan-Wu 
Chin}\authorrefmark{3}, and \uppercase{Duc-son Pham} \authorrefmark{2}
\IEEEmembership{Senior Member, IEEE}
}
\address[1]{Department of Informatics and Computer Engineering for Education, Sebelas Maret University, Surakarta 57126, Indonesia.}
\address[2]{School of Electrical Engineering, Computing, and Mathematical Sciences, Curtin University, Perth, WA 6102, Australia.}
\address[3]{School of Electrical, Computer, and Telecommunications Engineering, University of Wollongong, Wollongong, NSW 2500, Australia.}


\corresp{Corresponding author: Endar S. Wihidayat (e-mail: endars@staff.uns.ac.id).}

\begin{abstract}
In this paper, the Multi-stage Edge Server Upgrade (M-ESU) is proposed as a new network planning problem, involving the upgrading of an existing multi-access edge computing (MEC) system through multiple stages (e.g., over several years).
More precisely, the problem considers two key decisions: (i) whether to deploy additional edge servers or upgrade those already installed, and (ii) how tasks should be offloaded so that the average number of tasks that meet their delay requirement is maximized.
The framework specifically involves: (i) deployment of new servers combined with capacity upgrades for existing servers, and (ii) the optimal task offloading to maximize the average number of tasks with a delay requirement.
It also considers the following constraints: (i) budget per stage, (ii) server deployment and upgrade cost (in \$) and cost depreciation rate, (iii) computation resource of servers, (iv) number of tasks and their growth rate (in \%), and (v) the increase in task sizes and stricter delay requirements over time.
We present two solutions: a Mixed Integer Linear Programming (MILP) model and an efficient heuristic algorithm (M-ESU/H).  MILP yields the optimal solution for small networks, whereas M-ESU/H is used in large-scale networks.
For small networks, the simulation results show that the solution computed by M-ESU/H is within 1.25\% of the optimal solution while running several orders of magnitude faster. 
For large networks, M-ESU/H is compared against three alternative heuristic solutions that consider only server deployment, or giving priority to server deployment or upgrade. Our experiments show that M-ESU/H yields up to 21.57\% improvement in task satisfaction under identical budget and demand growth conditions, confirming its scalability and practical value for long-term MEC systems.
\end{abstract}

\begin{keywords}
Cloud Computing, Edge Server Placement, Network Planning/Upgrade, Task Delay, Task Offloading.  
\end{keywords}

\maketitle

\section{INTRODUCTION} \label{sec:introduction}
\IEEEPARstart{T}{he} growth of resource-demanding applications, such as those that make use of Artificial Intelligence (AI), has driven mobile devices to rely on external computing resources due to their limited resources\cite{dong2024task}.
However, the cloud may be located far from mobile devices, resulting in a longer time to transmit tasks to the cloud and receive corresponding results.  Furthermore, offloading tasks to the cloud may cause congestion or/and introduce security concerns\cite{zhang2024survey}.
Consequently, there has been increasing interest in Multi-access Edge Computing (MEC), which allows mobile devices to offload computation to edge servers collocated with access points (APs) or base stations.
This paradigm shift has led to many research studies, such as in \cite{shibata2025edge, ogawa2025transfer} and \cite{zeng2022parallel}. Further, the problem of optimal edge server placement remains a critical and complex research area \cite{asghari2024server}.

As the number of mobile devices and their computational requirements grow, e.g., to run AI services \cite{mao2024green}, edge servers have to be upgraded to meet the growing resource demands. In practice, (i) operators may deploy a limited number of edge servers, meaning not all APs have a server, and (ii) each deployed server may have limited resources, e.g., storage, computation, and bandwidth. 
One main reason for factor (i) and (ii) is due to the limited operator budget \cite{wu2025fairness}. In addition, some operators tend to deploy edge servers in stages to take advantage of newer technologies with improved resources and lower cost in the near future.
Factor (ii) is exacerbated by the growing number of mobile devices and resource-intensive applications, such as voice semantic analysis and face recognition~\cite{wang2019edge}, that have stringent deadlines.

To address factor (i), existing works such as \cite{wang2022optimal, li2021placement} deploy a limited number of edge servers at strategic locations; aka edge server placement problem \cite{wang2019edge}. 
However, after deployment, a network operator must also consider upgrading the resource or capacity of servers, e.g., equip them with Graphical Processing Units (GPUs), to cater for future growth in the number of tasks, end users, traffic volume and/or use of AI technologies \cite{xu2024unleashing}.  
Moreover, an operator has to consider tasks or applications with stringent delay requirements, e.g., \cite{nigade2021better}.

Accordingly, this paper examines a \textit{new} network planning problem termed Multi-stage Edge Server Upgrade (M-ESU).
The goal is to maximize the number of tasks meeting their delay requirements, called \textit{satisfied} tasks,  by deploying \textit{additional} edge servers or \textit{upgrading existing} ones over multiple stages (e.g., months or years).
M-ESU incorporates a set of constraints;
specifically:
(i) the maximum monetary budget available at each stage;
(ii) the initial cost of deploying new edge servers or upgrading existing ones, together with the rate at which these costs depreciate over time;
(iii) the transmission and processing delays incurred by tasks when offloaded to edge servers;
(iv) the computing capacity of edge servers; 
(v) the classification of tasks according to their delay tolerance;
specifically: (a) \textit{tolerant}, i.e., tasks that can be completed beyond their deadline, and (b) \textit{intolerant}, i.e., tasks that must be completed within their deadline,  
(vi) the number of tasks at each stage, which increases over time at a specified growth rate; and
(vii) the task size and/or the task delay requirement of some tasks may increase and become more stringent over stages, respectively. 
 \begin{figure*}[ht]
  \centering
 \includegraphics[width=0.8\textwidth, keepaspectratio]{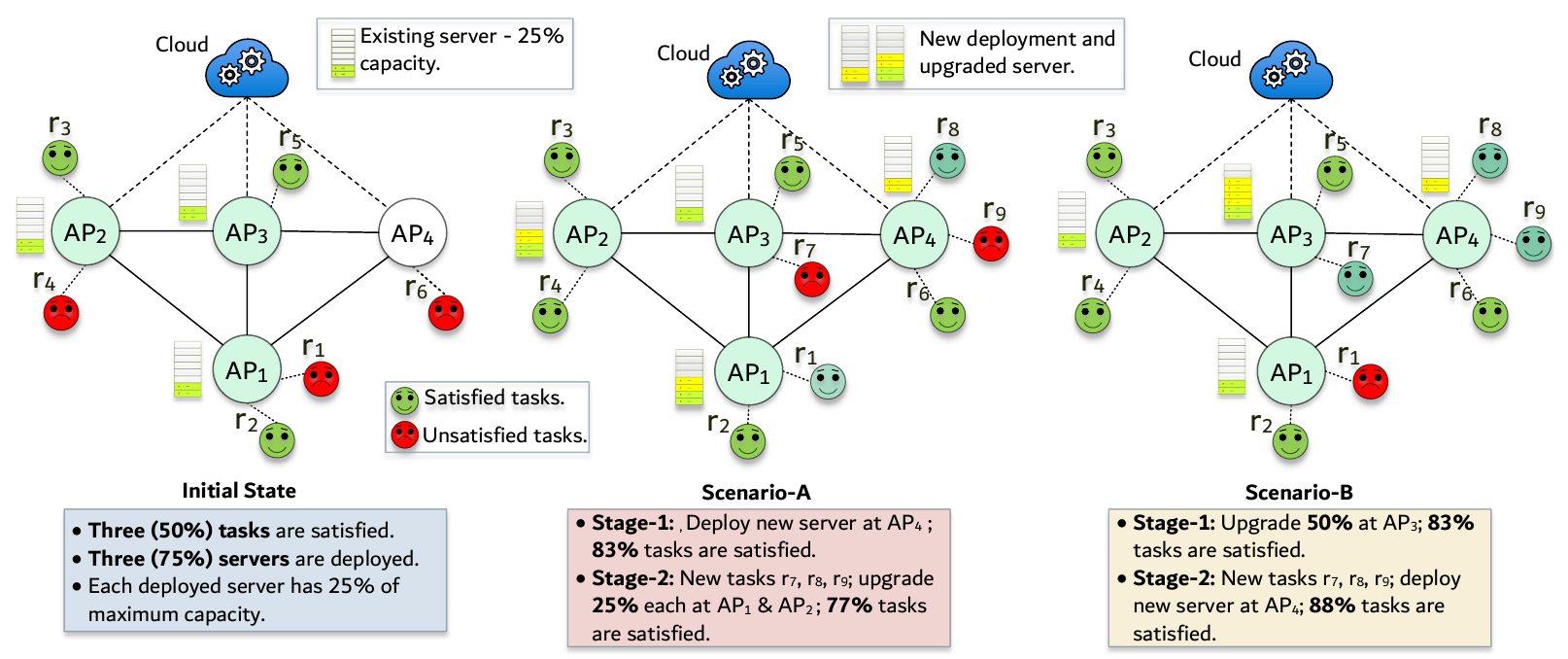}
  \caption{The impact of budget allocation and edge server upgrade on MEC network performance over two stages.}
  \label{fig_m_esu_example}
\end{figure*}

To illustrate M-ESU, consider the MEC network in Fig. \ref{fig_m_esu_example}(a), which consists of four APs ($AP_1$ to $AP_4$) and a cloud server. 
Edge servers are deployed at $AP_1$, $AP_2$, and $AP_3$, each operating at 25\% of their maximum capacity. Six tasks ($r_1$–$r_6$) are generated in the network, with three tasks (e.g., $r_3$ at $AP_2$) being satisfied, resulting in a task satisfaction rate of 50\% (three tasks are satisfied out of the total six tasks).
Two alternative upgrade scenarios are considered, each over two stages, as shown in Fig. \ref{fig_m_esu_example}(b)–(c). 
Both scenarios assume sufficient budget, where upgrading an existing server costs less than deploying a new one, and the number of tasks increases from six in Stage-1 to nine in Stage-2 with the addition of $r_7$–$r_9$.

In Scenario-A (Fig. \ref{fig_m_esu_example}(b)), the operator deploys a new server with 25\% capacity at $AP_4$ in Stage-1, and then upgrades the servers at $AP_1$ and $AP_2$ by an additional 25\% in Stage-2. 
This produces satisfaction rates of 83\% in Stage-1 and 77\% in Stage-2, resulting in an overall average satisfaction rate of 80\%.

In Scenario-B (Fig. \ref{fig_m_esu_example}(c)), the operator upgrades the server at $AP_2$ to 75\% capacity in Stage-1 and deploys a new server with 25\% capacity at $AP_4$ in Stage-2. This configuration yields 83\% and 88\% satisfied tasks in Stage-1 and Stage-2, respectively, for an average satisfaction rate of 85.5\%.

Scenario B achieves a higher satisfaction rate and possibly a lower overall budget, as deploying a new server in the second stage yields a greater reduction in deployment cost compared to performing upgrades. 
The results demonstrate that budget allocation, prioritizing new deployment or upgrade, and the selection of deployment/upgrade locations are critical factors influencing overall task satisfaction in multi-stage edge server planning.

We envisage the solution to M/ESU will be appealing to network operators for the following main reasons:
(i) spreading network upgrade across multiple stages offers greater budget flexibility compared to spending the whole budget at once. 
(ii) it allows operators to benefit from future technological advancements and cost reductions. Specifically, as technology prices are expected to decline over time, operators can acquire enhanced capabilities at a lower cost.
(iii) it gives operators two alternatives, but complementary network upgrades: deploying new edge servers or upgrading the capacity of existing servers. The first alternative is more costly and may not be viable due to budget constraints and/or limited available locations. On the other hand, the second alternative is more cost-effective, but is constrained by the maximum capacity of each server.

The M/ESU problem is challenging because it must determine the optimal budget allocation at each stage to decide the number and locations of new servers to deploy and/or the servers and their capacity to upgrade. In addition, it is essential to evaluate whether the deployment of new servers or the upgrade of server capacity at a specific location offers a more effective solution. 
Indeed, network upgrade problems that aim to determine which links or nodes to expand have been proven NP-hard \cite{paik1995network}.

This paper provides the following primary contributions:
\begin{enumerate}  
    \item It introduces a new network planning problem, referred to as M-ESU,
    which involves upgrading an existing MEC network over multiple stages via new edge servers and/or by increasing the capacity of existing edge servers. 
    The primary performance metric is the number of tasks that satisfy their delay requirements. 
    \item It introduces a Mixed Integer Linear Programming (MILP) formulation for solving M-ESU, which yields the optimal solution for small-scale networks.
    \item In addition, it proposes the first heuristic algorithm, called M-ESU/Heuristic (M-ESU/H), to solve M-ESU for large-scale networks. 
    \item It presents the first study that compares MILP and M-ESU/H. 
    Compared with MILP, M-ESU/H uses no more than 1\% of the CPU time while yielding, on average, only 1.25\% fewer satisfied tasks.
    In addition, M-ESU/H consistently delivers better results than three alternative heuristic solutions for M-ESU, called M-ESU/DO, M-ESU/DF, and M-ESU/UF. 
    M-ESU/DO involves server deployment only.  M-ESU/DF deploys new servers first before upgrading existing servers, whilst M-ESU/UF upgrades existing servers first before deploying new servers.
    The results indicate that M-ESU/H achieves, on average, 12.79\%, 13.99\%, and 21.57\% more satisfied tasks than M-ESU/DF, M-ESU/UF, and M-ESU/DO, respectively, under varying budgets, number of stages, number of tasks, and various cost ratios.
\end{enumerate}

The rest of this paper is organized as follows. Section \ref{sec_related_work} discusses related works, followed by Section \ref{sec_system_model}, which presents our system model. In Section \ref{sec_milp_solution}, we propose the MILP solution and Section \ref{sec_heuristic_solution} for our M-ESU/H solution. Section \ref{sec_evaluation} presents the performance of each proposed solution and
Section \ref{sec_conclusion} concludes the paper.

\section{RELATED WORK}\label{sec_related_work}
This section covers three MEC-related research topics: edge server placement, 
network capacity upgrade, and incremental or multi-stage deployment with budget constraints. 

A fundamental research problem in prior MEC research is edge server placement (ESP) to improve workload balance, system utilization, and latency.
Tiwari et al. \cite{tiwari2024knapsack} addressed the ESP problem in 5G networks with heterogeneous server capacities, formulating the problem as a 0–1 Knapsack optimization and solving it using a particle swarm optimization–based metaheuristic. Their approach achieved improved workload balance, utilization, and energy efficiency compared to baseline methods on real 5G datasets. 
Song et al. \cite{song2022delay} and  Yang et al. \cite{yang2019cloudlet} exemplify studies that integrate placement strategies with offloading strategies.   Briefly, in \cite{song2022delay}, the authors introducing an ant colony optimization–based task offloading algorithm that supports task partitioning and parallel execution across base stations, with a delay model that includes communication, execution, and migration delays.
On the other hand, the work in \cite{yang2019cloudlet} aims to minimize energy consumption on a small-scale edge server while meeting delay requirements using a Benders decomposition approach. 
Their model also considers divisible tasks and joint placement of edge servers and task allocation. Unlike the placement–offloading approaches above, Wang et al. in \cite{wang2020joint} concentrate solely on offloading and develop a joint optimization of down link resource usage, offloading decisions, and computation allocation to minimize the total user cost. Their model also reflects practical constraints by assuming finite MEC-server resources and incorporating both delay and monetary components.

Recent studies have recognized that edge server capacities available in MEC may not be sufficient to serve the increasing computational demand. 
There are several strategies to address this issue. Lähderanta et al. \cite{lahderanta2021edge} proposed deploying edge servers based on the needs of a specific cluster or area to improve server utilization. Another line of work focused on offloading more tasks to the cloud, although this carries the risk of increasing the number of unsatisfied tasks \cite{wihidayat2025multi}. 
A different approach relied on task migration, where tasks are redirected to other edge servers; however, this introduces additional transmission delays and can likewise raise the number of unsatisfied tasks. 
Fine-grained task allocation strategies have also been explored to enhance server utilization \cite{wihidayat2025multi, song2022delay, yang2019cloudlet}. Nevertheless, all of these strategies become ineffective when the total available edge server capacity is insufficient to meet the overall task demand.
Thus, one needs to eventually increase the resource capacity of the MEC initially deployed to meet ever-increasing user demands.   

The resource capacity of an existing MEC can be scaled up in two complementary ways: (i) by deploying additional edge servers, called horizontal scaling, and/or (ii) by increasing the capacity of the already deployed servers, called vertical scaling. Loven et al. \cite{loven2020scaling} investigated horizontal scaling strategies in which additional servers are deployed at new locations to increase network capacity. 
Similar to \cite{loven2020scaling}, Niu et al. \cite{niu2025esd} introduce A–ESD, expanding the network capacity by deploying auxiliary edge servers in between existing edge servers. The auxiliary servers essentially become secondary edge nodes with smaller capacity and communication range to assist primary servers.
These studies confirm that expanding deployment footprints improves coverage and QoS, but they do not explore upgrading the capacity of already-deployed servers.
On the other hand, recently, Xu et al. \cite{xu2025budget} proposed the Budget-Constrained Edge Server Expansion Deployment (BC-ESED) framework, which incorporates both horizontal scaling and vertical scaling. Their approach introduces a paradigm that allows not only the deployment of new servers but also the expansion of the capacities of existing ones. Formulated as a multi-objective optimization problem, BC-ESED simultaneously minimizes the average access delay and the workload deviation under budget constraints, which is solved using a GA–PSO–based metaheuristic.

Two prior studies have considered placement, deployment, or network upgrade over a single stage.  
Ren et al. \cite{ren2021demand} proposed a Demand-Driven Incremental Deployment (DDID) strategy for IoT edge computing, where server placement is optimized over multiple stages or periods according to service demand classified as rigid or non-rigid. 
Their approach formulates deployment as a bi-level optimization problem, solved with subgradient optimization and evolutionary algorithms, and demonstrates that incremental deployment can reduce costs compared to one-time placement. 
Furthermore, Wihidayat et al. \cite{wihidayat2025multi} formulated M-ESD as a multi-stage edge server deployment problem, beginning with the deployment of new servers in the initial stage before expanding network capacity with additional servers in later stages.  Their objective 
to maximize the average number of satisfied tasks across multiple stages,
where their proposed solutions incorporate a fine-grained offloading strategy with divisible tasks that enables higher server utilization under limited budget constraints.

Although existing MEC works have covered important aspects of edge server placement, task offloading, upgrading, and incremental deployment, they each leave critical gaps. 
Table~\ref{tbl_related_works} shows these gaps.
In summary, placement and offloading studies (e.g., \cite{tiwari2024knapsack,song2022delay,yang2019cloudlet}) confirm the benefits of optimized server locations and divisible task allocation, but they do not consider long-term capacity planning or budget constraints. 
Upgrade-oriented works (e.g., \cite{loven2020scaling,niu2025esd,xu2025budget}) demonstrate that horizontal or vertical scaling can reduce delay and improve QoS, though they are typically limited to single-stage strategies and short-term scheduling. 
Incremental deployment studies (e.g., \cite{ren2021demand,wihidayat2025multi}) extend this view to multi-stage settings, but they often restrict the scope to \textit{new} deployments only and overlooking upgrading the capacity of existing servers. 

Taken together, no works integrate multi-stage deployment and server capacity upgrades under budget allocation, while also accounting for task growth, increasingly stringent deadlines required by future applications, e.g., AI, and growing task sizes over successive stages, and fractional offloading across edge and cloud resources.
By contrast, our proposed M-ESU framework addresses these gaps by combining these dimensions into a comprehensive model, bridging the disconnect between placement, offloading, upgrading, and multi-stage planning in MEC research.
\begin{table}[]
\caption{\textbf{Comparison of Relevant Works}.}
\label{tbl_related_works}
\begin{tabular}{lcccccccc}
\hline
\multicolumn{1}{c}{\textbf{Ref.}}                                                   & \textbf{PL}             & \multicolumn{1}{l}{\textbf{HS}}             & \textbf{VS}             & \textbf{MT}             & \textbf{CS}             & \textbf{BG}             & \textbf{TS}             & \textbf{FC}             \\ \hline
\cite{tiwari2024knapsack,lahderanta2021edge}                                                           & \cmark                  & \xmark                                      & \xmark                  & \xmark                  & \xmark                  & \xmark                  & \xmark                  & \xmark                  \\
\cite{song2022delay}                                                                & \cmark                  & \xmark                                      & \xmark                  & \xmark                  & \xmark                  & \xmark                  & \xmark                  & \cmark                  \\
\cite{yang2019cloudlet}                                                             & \cmark                  & \xmark                                      & \xmark                  & \xmark                  & \xmark                  & \xmark                  & \xmark                  & $\sim$                  \\
\cite{loven2020scaling}, \cite{niu2025esd}                                          & \cmark                  & \cmark                                      & \xmark                  & \xmark                  & \xmark                  & \xmark                  & \xmark                  & \xmark                  \\
\cite{xu2025budget}                                                                 & \cmark                  & \cmark                                      & \cmark                  & \xmark                  & $\sim$                  & \cmark                  & \xmark                  & \xmark                  \\
\cite{ren2021demand}                                                                & \cmark                  & \xmark                                      & \xmark                  & \cmark                  & \xmark                  & $\sim$                  & \xmark                  & \xmark                  \\
\cite{wihidayat2025multi}                                                           & \cmark                  & $\sim$                                      & \xmark                  & \cmark                  & \cmark                  & \cmark                  & \cmark                  & \cmark                  \\ \hline
\multicolumn{1}{c}{\textbf{Our work}}              & \cmark                  & \cmark                                      & \cmark                  & \cmark                  & \cmark                  & \cmark                  & \cmark                  & \cmark                  \\ \hline
\multicolumn{9}{c}{\scriptsize
{\begin{tabular}[c]{@{}c@{}}
\cmark = full coverage, \xmark = not covered, $\sim$ = partial coverage, PL = placement,\\ HS = horizontal scaling,  VS = vertical scaling, MT = multistage, \\ CS = dynamic cost, BG = budget, TS = task growth \& stringent, FC = fraction.\end{tabular}}}
\end{tabular}
\end{table}

\section{SYSTEM MODEL}\label{sec_system_model}
This section first describes the network model, followed by the model for servers, tasks, delays, upgrade and maintenance cost and budget. 
\begin{table}[htbp]
\caption{\textbf{Summary of notations}.}
\label{tab: tbl_notation}
\setlength{\tabcolsep}{3pt}
\small
\begin{tabular}{|p{50pt}|p{185pt}|}
\hline
\textbf{1} & \textbf{Network}\\
\hline
    $T$ & Total number of planning stages.\\
    $G^t(\mathcal{V},\mathcal{E},\mathcal{S})$ & The MEC network at stage $t$ with $|\mathcal{V}|$ APs, $|\mathcal{E}|$ logical links, and $|\mathcal{S}|$ edge servers.\\ 
    $\mathcal{P}_{u,v}$ & A path from an AP $u$ to an AP $v$.\\
    $\pi$ & The propagation rate of a medium.\\
    $l_{i,j}$ & The length of link  $(i,j)$.\\
    $d_{u,v}$ & 
    The physical distance of path $\mathcal{P}_{u,v}$.\\
    $\eta_{i,j}$ & The transmission rate of link $(i,j)$.\\
\hline 
\textbf{2} & \textbf{Task}\\
\hline
    $r_k^t$ & A task $k$ at stage $t$. \\
    $r^t_{k,s}$ & A fraction of task $r_k^t$ that is offloaded to server $s$.\\
    $\mathcal{F}^t_k$ & The set of all fractions of task $r_k^t$.\\
    $b_k^t$ & The size of task $r_k^t$.\\
    $b_{k,s}^t$ & The size of task fraction $r^t_{k,s}$.\\
    $\tau_k^t$ & The delay of task $r_k^t$.\\
    $\tau_{k,s}^t$ & The delay of task fraction $r^t_{k,s}$.\\
    $\mathcal{R}^t$ & A set of all tasks at stage $t$.\\
    $\mu$ & The increase rate of the number of tasks per stage.\\
    $\rho_b (\rho_\tau)$ & 
    The percentage of tasks with increased size (stringent deadline).\\
    $\alpha_b (\alpha_\tau)$ & The increase (tightened) rate of task size (deadline).\\
    $\zeta$ & The size multiplier of the result of each task $r_k^t$.\\
    $\Gamma^t$ & The total number of \textit{satisfied} tasks at stage $t$.\\
    $\bar{\Gamma}$ & The average number of \textit{satisfied} tasks over $T$ stages. \\
\hline
\textbf{3} & \textbf{Server}\\
\hline
    $\mathcal{S}^t$ & A set of deployed edge servers from stage 1 to $t$ .\\
    $C_p$ & The computational capacity of each rpack.\\
    $M^t_s(m^t_s)$ & 
    The total number of rpacks installed from stages 1 to $t$ (at stage $t$). \\
    $\psi_s$ & The maximum capacity of each server $s$.\\ 
    $\hat{\psi}_s^t(\psi_s^t)$ & The workload (residual capacity) of server $s$ at stage $t$.\\
    $\beta_{s}$ & The processing speed of server $s$.\\
    $\theta^t_{k,s}$ & The processing time of server $s$ to execute task fraction $r_{k,s}^t$.\\
\hline
\textbf{4} & \textbf{Cost \& Budget}\\
\hline
$B$($B^{t}$) & The allocated budget over $T$ stages (at stage $t$). \\
    $\kappa_n^t$ & The cost to deploy an edge server at stage $t$.\\
    $\kappa_u^t$ & The cost to install an rpack at stage $t$. \\
    $I^t$ & The cost to deploy a server  infrastructure at stage $t$.\\
    $\phi$ & The depreciation rate of the cost of each server and its infrastructure and rpack per stage. \\
\hline
\end{tabular}
\end{table}
%

\subsection{NETWORK MODEL}\label{sec_network_model}
Let $T \geq 1$ specify the horizon within which the operator seeks to upgrade the network. For each \textit{time} or \textit{stage} $t \le T$, its duration follows the operator’s budget allocation cycle—typically an annual process influenced by anticipated increases in task demand~\cite{ren2021demand}.
The network at stage $t \in \{1,2,\ldots,T\}$ is denoted by $G^t(\mathcal{V}, \mathcal{E}, \mathcal{S})$. Here, $\mathcal{V}$ includes $|\mathcal{V}|-1$ APs and a cloud server $c$, while $\mathcal{S}$ contains the edge servers along with the cloud $c$. 
We refer to the initial network as $G^0(\mathcal{V}, \mathcal{E}, \mathcal{S})$, which includes $1 \leq |\mathcal{S}| \leq |\mathcal{V}|$ servers; the cloud server $c$ within this set is assumed to have unlimited capacity.
The set $\mathcal{E}$ includes logical links, with each link corresponding to a channel connecting either between a pair of APs or between an AP and the cloud $c$. Each logical link can encompass one or several underlying physical links.
%
\subsection{SERVER MODEL} \label{sec_server_capacity}
An edge server $s \in \mathcal{S}$ is placed at a corresponding AP $u \in \mathcal{V}$; moreover, no AP hosts more than one edge server. At each stage $t$, an operator may (i) deploy \textit{new} servers and/or (ii) upgrade \textit{existing} servers.
We consider that each server is upgraded by installing at least one \textit{resource-upgrade pack}, or \textit{rpack} in short. Without loss of generality, we consider each newly deployed server to include at least one rpack. 
Let $M_{max} \geq 1$ be the maximum number of rpacks that can be installed for each server. Thus, a network $G(\mathcal{V},\mathcal{E},\mathcal{S})$ can have up to $\mathcal{M} = (\mathcal{|V|}-1) \times M_{max}$ rpacks. 

Each rpack has a computational capacity $C_p$ (in bits), which includes sufficient supporting facilities, such as CPU, memory, storage, and communication.
Let $m^t_s$ and $M^t_s \leq M_{max}$ denote the number of rpacks in server $s$ that are installed at stage $t$ and installed from stage 1 to $t$, respectively. Thus, we have $M^t_s = \sum_{\tau=1}^t m^{\tau}_s$.
The maximum capacity of server $s$, denoted by $\psi_s$, is $\psi_s = C_p \times M_{max}$.
At each stage $t$, let $ \hat{\psi}_s^t $ represent the \textit{workload} of server $s$, for $\hat{\psi}_s^t \leq (C_p \times M^t_s)$.  
We denote the \textit{residual} capacity of server $s$ at stage $t$ by $\psi^t_s$, given as
\begin{equation}\label{eq_remaining_capacity}
\psi_s^t = (C_p \times M^t_s) - \hat{\psi}_s^t, \quad \forall s \in \mathcal{S}, \forall t \leq T.
\end{equation}
%

\subsection{TASK MODEL}\label{sec_task_req_model}
We represent the $k$-th task coming from AP $u$ during stage~$t$ as 
$r_k^t = (u, \tau_{max,k}^t, b_k^t, \sigma_k^t)$, which specifies the maximum delay requirement $\tau_{max,k}^t$ (in seconds), task size $b_k^t$ (in bits), and delay tolerance $\sigma_k^t \geq 1.0$.
We denote the set of tasks originating from AP $u$ at stage $t$ by $\mathcal{R}^t_u = \{(u, \tau_{max,k}^t, b_k^t, \sigma^t_k)~| u \in \mathcal{V}\}$. Consequently, the aggregate set of tasks $\mathcal{R}^t$ is defined as the union $\mathcal{R}^t = \bigcup_{u \in \mathcal{V}}\mathcal{R}^t_u$.
We model the growth in the total number of tasks using a per-stage rate $\mu \geq 0$. Formally,
\begin{equation}{\label{eq_tasks_increase}}
      |\mathcal{R}^{t}| = (1 + \mu)^{t-1} \cdot |\mathcal{R}^{1}|.
\end{equation}
Over time, as shown in \cite{badnava2023energy, song2022delay}, task sizes have been observed to increase, while delay constraints become more stringent. To reflect these trends in our model, we introduce $\rho_b \in [0,1]$ as the proportion of tasks in each node whose size increases and $\rho_\tau \in [0,1]$ the proportion of tasks at each node whose delay constraint becomes tighter per stage. 
Further, we use $\alpha_b \geq 0$ to denote the task size growth rate, and $\alpha_\tau \in [0,1]$ to denote the deadline tightening rate. 
Formally, for a task $r_k^t = (u, \tau_{max,k}^t, b_k^t, \sigma^t_k)$, we have
\begin{equation}\label{eq_task_size_stringent}
b_k^t = (1 + \alpha_b)^{t-1} \cdot b_k^1,
\end{equation}
\begin{equation}\label{eq_task_deadline_stringent}
\tau_{max,k}^t = (1 - \alpha_\tau)^{t-1} \cdot \tau_{max,k}^1.
\end{equation}

Each task $r_k^t$ can be categorized into a \textit{intolerant} or \textit{tolerant} task. 
Let $\tau^t_k$ represent the completion time of task $r_k$ at stage $t$, and $\sigma^t_k \geq 1.0$ be its tolerance multiplier. 
A task $r_k^t$ is considered \textit{satisfied} if its completion time does not exceed $\sigma^t_k \times \tau_{max,k}^t$, i.e., $\tau^t_k \leq \sigma^t_k \times \tau_{max,k}^t$. 
At each stage $t$, an \textit{intolerant} task requires $\sigma^t_k = 1$, while each \textit{tolerant} task $r_k$ is satisfied as long as the delay remains within $\sigma^t_k \times \tau_{max,k}^t$, for  $\sigma^t_k > 1.0$.
Consistent with \cite{yang2019cloudlet, li2018computation}, we allow each task $r_k^t$ to be segmented into multiple \textit{fractions} that can be distributed across distinct servers.
Let $\mathcal{F}^t_k$ denote the set of fractions of task $r_k^t$. The fraction assigned to server $s$ is represented by $r^t_{k,s}$ with the size $b^t_{k,s}$ (in bits), for $s = 1, 2, \cdots$,$| \mathcal{S}^t|$. 
Consequently, the total size of task $r^t_k$ at stage $t$ is given by
\begin{equation}{\label{eq_task_size_from_fraction}}
    b_{k}^t = \sum_{r^t_{k,s} \in \mathcal{F}^t_k} b^t_{k,s}.
\end{equation}
Since each server $r_k^t$ can process at most one fraction of each task, the number of fractions per task satisfies $|\mathcal{F}^t_k| \leq |\mathcal{S}^t|$.
The processing delay of a fraction $r_{k,s}^t$ is given by
\begin{equation}{\label{eq_processing_delay}}
\theta_{k,s}^t = b_{k,s}^t~ /~ \beta_s,
\end{equation}
where $\beta_s$ represents the processing rate of server $s$.

For each stage $t$, three allocation strategies are considered: 
(i) offloading all fractions to the cloud $c$;
(ii) offloading all fractions to one or more edge servers;
or (iii) distributing fractions between edge servers and the cloud.
After processing, each server $s$ transmits the \textit{result} of fraction $r^t_{k,s}$ back to the originating node of task $r_k^t$ (i.e. AP $u$). 
The size of the result, denoted by $\bar{b}^t_{k,s}$, is proportional to the size of task fraction and is computed as
\begin{equation}{\label{eq_task_size_result}}
    \bar{b}^t_{k,s} = \zeta \times b^t_{k,s},
\end{equation}
where $\zeta \in [0, 1.0]$ is the result size multiplier.

\subsection{DELAY MODEL}\label{sec_delay_model}
We define $\beta_s$ as the processing rate of server $s$, expressed in CPU cycles per second \cite{yang2019cloudlet}.
The processing speeds of the edge servers and the cloud are assumed to be constant.
The transmission rate of a logical link $(i, j)$ is denoted by $\eta_{i,j}$ (in bits per second), and its physical length is $I_{i,j}$ (in meters). 
We denote the routing path from AP $u$ to AP $v$ by $\mathcal{P}_{u,v}$, for which the end-to-end delay comprises both \textit{propagation} and \textit{transmission} components.

The total distance of the path $\mathcal{P}_{u,v}$ is expressed as 
\begin{equation}{\label{eq_distance_propagation_delay}}
    d_{u,v} = \sum_{(i,j) \in \mathcal{P}_{u,v}} l_{i,j}, 
\end{equation}
where $l_{i,j}$ is the distance (in meters) between nodes $i$ and $j$.

Given a propagation speed $\pi$ (in meters per second), the propagation delay from node $u$ to $v$ is
\begin{equation}{\label{eq_propagation_delay}}
    \delta^{'}_{u,v}(\pi) = d_{u,v} / \pi.
\end{equation}
The transmission delay for sending a rpack of size $b$ bits along a path $\mathcal{P}_{u,v}$ is given by
\begin{equation}{\label{eq_transmission_delay}}
    \delta^{''}_{u,v}(b) = \sum_{\forall(i,j)\in\mathcal{P}_{u,v}}b~/~\eta_{i,j}. 
\end{equation}
Hence, the routing delay for a rpack of size $b$ along $\mathcal{P}_{u,v}$ is
\begin{equation}{\label{eq_routing_delay}}
     \delta_{u,v}(\pi,b) = \delta^{'}_{u,v}(\pi) + \delta^{''}_{u,v}(b).
\end{equation}

Let $\tau^t_{k,s}$ (in seconds) denote the completion delay of a task fraction $r^t_{k,s}$ at stage $t$.
Following \cite{cruz2022edge}, this delay is the sum of three components: (i) the routing delay $\delta^t_{u, s}(\pi,b^t_{k,s})$ to transmit  fraction $r^t_{k,s}$ from the originating AP $u$ to the server $s$;
(ii) the processing delay $\theta_{k,s}^t$ incurred on the server $s$; and 
(iii) the routing delay $\delta^t_{s, u}(\pi,\bar{b}^t_{k,s})$ required to return the processed result of fraction $r^t_{k,s}$ from the server $s$ back to $u$. 
Consistent with \cite{yang2019cloudlet}, the transmission time between users and APs is neglected. 
%
Consequently, the total completion delay for fraction $r^t_{k,s}$ is formulated as
\begin{equation}{\label{eq_task_fraction_delay}}
    \tau^t_{k,s} 
    = \delta_{u,s} (\pi,b^t_{k,s}) + \theta^t_{k,s} +  \delta_{s,u}(\pi,\bar{b}^t_{k,s}), 
\end{equation}
and the delay of task $r_k^t$ is determined by the maximum completion delay among all its fractions, i.e.,
\begin{equation}
\tau^t_k = \max\{
\tau^t_{k,s}\}, s \in \mathcal{S}^t, \forall r^t_{k,s} \in \mathcal{F}^t_k. 
\label{eq_task_delay}
\end{equation}
%

\subsection{COST AND BUDGET MODEL}\label{sec_cost_budget_model}
Following  \cite{ren2021demand}, the operator has costs that consist of (i) capital expenditure (CAPEX), and (ii) operating expenditure (OPEX).  
We consider that CAPEX includes (a) the cost to \textit{initially} deploy the server, and (b) the cost to install each rpack. 
OPEX, in contrast, represents the cost associated with operating and maintaining the server, including its energy consumption.
Let $\kappa_u^t$ be the cost to install one rpack in stage $t$, and $\kappa_n^t$ be the cost to deploy a new server in stage $t$. 
The latter includes the cost to install the required \textit{infrastructure} (e.g., computer rack and other supporting peripherals) for the server at stage $t$, denoted by $I^t$.
We consider that each new server includes \textit{at least} one upgrade rpack. Thus, we have 
\begin{equation}{\label{server_deployment_cost}}
\kappa_n^t = I^t + m^t_s \cdot \kappa_u^t,  
\end{equation}
for $1 \leq m^t_s \leq M_{max}$. 
Recall that $m^t_s$ is the number of rpacks that are installed in the server $s$ at stage $t$, and $M_{max}$ is the maximum number of rpacks that can be installed in a server. 
 
We expect the CAPEX and OPEX of edge server deployment to decline in subsequent stages, owing to vendor competition and technology maturity \cite{ren2021demand}.
And we define $\phi$ as the cost depreciation rate per stage for new server deployments and rpack installations, where $0 \leq \phi < 1$.
Thus, we have 
\begin{equation}{\label{eqeq_deployment_cost}}
    I^t = (1 - \phi)^{t-1} \cdot I^1,
\end{equation} 
and 
\begin{equation}{\label{eq_deployment_cost}}
    \kappa_u^t = (1 - \phi)^{t-1} \cdot \kappa_u^1.
\end{equation}

We denote the total budget over $T$ stages by $B$, yielding an allocated budget of $B^t = B/T$ for each stage $t$. Any unspent budget at stage $t$, denoted by $\Delta B^t$, carries over to augment the budget of the next stage $t+1$.
Thus, we have 
\begin{equation}{\label{eq_budget_perstage}}
    B^{t} = B^{t} + \Delta B^{t-1},
\end{equation}
for $t = 1, 2, \dots, T$ and $\Delta B^{0} = 0$.

\section{PROBLEM FORMULATION}\label{sec_milp_solution}
The Multi-stage Edge Server Upgrade (M-ESU) problem is formulated with the objective of maximizing the average number of satisfied tasks, $\Bar{\Gamma}$,
by \textit{upgrading} the MEC network through selective actions: (i) increasing the capacity of existing edge servers and/or (ii) deploying additional edge servers.
At any stage $t$, we employ the binary variable $\gamma^{t}_k \in \{0,1\}$ to signify the status of task $r_k^t$: satisfied if $\gamma^{t}_k=1$ and unsatisfied if $\gamma^{t}_k=0$. Accordingly, the aggregate number of satisfied tasks $\Gamma^t$ for stage $t \leq T$ is expressed as
\begin{equation}{\label{eq_number_satisfied_task}}
\Gamma^t = \sum_{\forall r_k^t \in \mathcal{R}^t} \gamma^t_k.
\end{equation}

The M-ESU problem aims to maximize the average number of satisfied tasks over all stages, denoted by $\bar{\Gamma}$. Formally, the objective function is given by
\begin{equation}{\label{eq_milp_average_satisfied_task}}
\text{Max } \bar{\Gamma} = \frac{1}{T} \sum_{t=1}^{T} \Gamma^t.
\end{equation}

\subsection{CONSTRAINTS}\label{sec_constraints}
The constraints within M\mbox{-}ESU cover budget restrictions, task\mbox{-}fraction handling, capacity limits of servers, deployment and upgrade conditions, and delay requirements for tasks.
\subsubsection{Budget} 
We use $d^t_s$ as a binary variable, which is assigned a value of one whenever server $s$ is deployed at stage~$t$.
Constraint (\ref{eq_budget_constraint}) guarantees that, for each stage 
$t \leq T$, the total upgrade cost remains within the available budget at that stage.
The left-hand side of the constraint represents the total cost of deploying new servers and upgrading existing ones at stage $t$.
while the right-hand side denotes the cumulative budget available up to stage $t$ \textit{minus} the total upgrade cost incurred up to stage $t-1$, i.e., the budget that may be used at stage $t$.

\begin{equation}{\label{eq_budget_constraint}}
\begin{split}
\sum_{s \in \mathcal{S}^t}(\kappa_n^t \times d^t_s + \kappa_u^t \times m^t_s) \leq \sum_{k=1}^{t}B^k - 
\\
\sum_{k=1}^{t-1} \sum_{s \in \mathcal{S}^k} (\kappa_n^k \times  d^k_s + \kappa_u^k \times m^k_s ), \quad \forall t \leq T.
\end{split}
\end{equation}

\subsubsection{Task}
Constraint (\ref{eq_fractions_constraint}) ensures that, for each stage $t$, the aggregate size of all fractions belonging to task $r_k^t$ distributed across one or more servers equals the original task size $b_k$. 
Recall that $b^t_{k,s}$ represents the size of the task fraction $r^t_{k,s}$.
\begin{equation} {\label{eq_fractions_constraint}}
    \sum_{r_{k,s} \in \mathcal{F}^t_k} b^t_{k, s} = b_k, \quad \forall r_k \in \mathcal{R}^t, \forall t \leq T.
\end{equation}
For each fraction of tasks, constraint (\ref{eq_cloud_fractions_constraint}) aims to offload the maximum possible size to the cloud, in which the delay remains within the task deadline. i.e., $b_{u,k}^t$ from Eq. (\ref{eq_max_fraction}).
\begin{equation} {\label{eq_cloud_fractions_constraint}}
    b_{u,k}^t = 
    \begin{cases}
    b_k^t, & \text{if } b_{u,k}^t \geq b_k^t, \\
    b_{u,k}^t, & \text{otherwise,} \quad \forall r_k \in \mathcal{R}^t, \forall t \leq T.
    \end{cases}
\end{equation}
\subsubsection{Server Capacity}

Each server is required to operate within its capacity.
Constraint~(\ref{eq_capacity_constraint}) restricts the aggregate size of task fractions offloaded to server~$s$ remains within its capacity, given by $m^t_s \times C_p$:
\begin{equation}{\label{eq_capacity_constraint}} 
\sum_{r_k^t \in \mathcal{R}^t} b^t_{k, s} \leq M^t_s \cdot C_p, \quad \forall s \in 
\{\mathcal{S}^t-c\}, 
\forall t \leq T.
\end{equation}

\subsubsection{Server Deployment}
A server must be deployed before any tasks can be offloaded to it. Accordingly, the following constraint is imposed:
\begin{equation}{\label{eq_deployment_constraint_1}}
    b^t_{k, s} \leq M^t_s \cdot C_p \cdot d^t_s, \quad \forall r_k^t \in \mathcal{R}^t, \forall s \in \mathcal{S}^t, \forall t \leq T, 
\end{equation}
where $d^t_s$ is a binary variable used to determine 
whether server~$s$ has an edge server deployed before a task fraction $r_{k,s}^t$ of size $b^t_{k,s}$ is offloaded to it.
Next, we impose that each node $s$ can host at most one edge server and that any deployment decision cannot be reversed.
Mathematically, we have
\begin{equation}{\label{eq_undeployed_constraint}}
d^t_s \geq d^{t-1}_s, \quad \forall t \leq T, \forall s \in \mathcal{V}.
\end{equation}
\subsubsection{Server Upgrade}
The server has a maximum number of rpacks that can be installed. 
Thus, we have the following constraint:
\begin{equation}\label{eq_max_upgrade_constraint}
    M^t_s \leq M_{max}, \quad \forall s \in \mathcal{S}^t, \forall t \in T.
\end{equation}
Further, once a server is upgraded, it cannot be downgraded or its capacity be transfered to another server in the next stage. 
Thus, we add the following constraint:
\begin{equation}\label{eq_no_downgrade_constraint}
    M^t_s \geq M^{t-1}_s, \quad \forall s \in \mathcal{S}^t, \forall t \in T.
\end{equation}
\subsubsection{Task Delay}
Constraint (\ref{eq_delay_constraint}) ensures that a task fraction is marked satisfied only if its delay does not exceed the deadline, 
where the constant $M$ is used to relax the condition when the deadline is violated.
\begin{equation} {\label{eq_delay_constraint}}
\begin{split}
\tau^t_{k, s} \leq \tau_{max,k}^t + (1-\gamma^t_{k, s}) \cdot M,
\quad \forall s \in \mathcal{S}^t, 
\\
\forall r_k^t \in \mathcal{R}^t, \forall r_{k,s}^t \in \mathcal{F}^t_k, \forall t \leq T.
\end{split}
\end{equation}
Let $\gamma^t_{k,s}$ be a binary variable that is set to one if a task fraction $r_{k,s}$ is satisfied and to zero otherwise. The following constraint (\ref{eq_satisfied_constraint}) requires that, at any stage $t$, a task $r_k$ is satisfied only when each of its fractions $r_{k,s}$ meets its deadline, i.e.,$\gamma^t_{k,s} = 1$:
\begin{equation}{\label{eq_satisfied_constraint}}
    \gamma^t_k \leq \gamma^t_{k,s}, \quad \forall s \in \mathcal{S}^t,\forall r_k^t \in \mathcal{R}^t, \forall r_{k,s}^t \in \mathcal{F}^t_k, \forall t \leq T.
\end{equation}

\subsection{PROBLEM COMPLEXITY}
The computational complexity of MILP (\ref{eq_milp_average_satisfied_task}) is determined by the number of constraints and variables.  In this respect, we have the following proposition.
\begin{prop}\label{prop_time_complexity_MILP_MESU}
MILP~(\ref{eq_milp_average_satisfied_task}) has
$|\mathcal{V}|T+\sum_{t=1}^T (|\mathcal{S}^t| + 2 (|\mathcal{S}^t||\mathcal{R}^t| + |\mathcal{R}^t|))$
decision variables, and
$T + |\mathcal{V}|T+\sum_{t=1}^T(2|\mathcal{R}^t| + 3|\mathcal{S}^t| 
+ 3|\mathcal{S}^t||\mathcal{R}^t|-1)$
constraints.
\end{prop}

\begin{proof}
MILP (\ref{eq_milp_average_satisfied_task}) includes six types of decision variables, i.e., 
$d^t_s$, $m^t_s$, $b^t_{k,s}$, $b^t_{u,k}$, $\gamma^t_k$, and $\gamma^t_{k,s}$. 
In (\ref{eq_budget_constraint}), 
there are $|\mathcal{V}|T$ variables of type $d^t_s$ because it considers the number of nodes at each stage $t$.
Next, we have $\sum_{t=1}^T |\mathcal{S}^t|$ variables of type $m^t_s$ in (\ref{eq_budget_constraint}) for server deployment and upgrade decisions, respectively. 
Variables $\gamma^t_k$ in (\ref{eq_number_satisfied_task}) and $\gamma^t_{k,s}$ in (\ref{eq_satisfied_constraint}) account for task satisfaction at stage $t$ and for each task–server pair, 
resulting in $\sum_{t=1}^T |\mathcal{R}^t|$ and $\sum_{t=1}^T |\mathcal{S}^t||\mathcal{R}^t|$ variables, respectively. 
The offloading-related variables $b^t_{k,s}$ in (\ref{eq_task_size_from_fraction}) and $b^t_{u,k}$ (\ref{eq_cloud_fractions_constraint}) contribute 
$\sum_{t=1}^T |\mathcal{S}^t||\mathcal{R}^t|$ and $\sum_{t=1}^T |\mathcal{R}^t|$ variables, respectively. 
Hence, MILP~(\ref{eq_milp_average_satisfied_task}) contains in total 
$
|\mathcal{V}|T+\sum_{t=1}^T (|\mathcal{S}^t| + 2 (|\mathcal{S}^t||\mathcal{R}^t| + |\mathcal{R}^t|))
$
decision variables.

Next, MILP (\ref{eq_milp_average_satisfied_task}) contains ten types of constraints, namely  
(\ref{eq_budget_constraint}), (\ref{eq_fractions_constraint}), (\ref{eq_cloud_fractions_constraint}), (\ref{eq_capacity_constraint}), (\ref{eq_deployment_constraint_1}), (\ref{eq_undeployed_constraint}), (\ref{eq_max_upgrade_constraint}), (\ref{eq_no_downgrade_constraint}), (\ref{eq_delay_constraint}), and (\ref{eq_satisfied_constraint}).  
Constraint (\ref{eq_budget_constraint}) limits the total deployment and upgrade cost at each stage, resulting in $T$ constraints.  
Constraints (\ref{eq_fractions_constraint}) and (\ref{eq_cloud_fractions_constraint}) 
use
$2\sum_{t=1}^T|\mathcal{R}^t|$ constraints.  
Constraints (\ref{eq_capacity_constraint})–(\ref{eq_deployment_constraint_1}) 
account for server capacity and task deployment, contributing 
$\sum_{t=1}^T (|\mathcal{S}^t|-1)$ and $\sum_{t=1}^T|\mathcal{S}^t||\mathcal{R}^t|$ constraints, respectively. 
Constraint (\ref{eq_undeployed_constraint}) maintains deployment monotonicity across stages, 
adding $|\mathcal{V}|T$ constraints. 
Constraints (\ref{eq_max_upgrade_constraint})–(\ref{eq_no_downgrade_constraint}) 
limit the number of installed rpacks per server and enforce non-decreasing upgrades, 
yielding $2\sum_{t=1}^T|\mathcal{S}^t|$ constraints. 
Constraint (\ref{eq_delay_constraint}) relates each task fraction to its deadline, 
producing $\sum_{t=1}^T|\mathcal{S}^t||\mathcal{R}^t|$ constraints. 
Finally, constraint (\ref{eq_satisfied_constraint}) defines task satisfaction based on all its fractions, 
adding another $\sum_{t=1}^T|\mathcal{S}^t||\mathcal{R}^t|$ constraints.
Summing all these together, MILP (\ref{eq_milp_average_satisfied_task}) contains  
$T + |\mathcal{S}|T+\sum_{t=1}^T(2|\mathcal{R}^t| + 3|\mathcal{S}^t| 
+ 3|\mathcal{S}^t||\mathcal{R}^t|-1)$
constraints in total.
\end{proof}

\section{HEURISTIC SOLUTION}\label{sec_heuristic_solution}
This section presents M-ESU/H. 
In Section~\ref{sec_m_esu_h_overview}, we outline the M-ESU/H strategy for selecting the node on which to deploy an edge server and for identifying the server to which tasks are offloaded.
Section~\ref{sec_detail_m_esu_h} presents the detailed description of M-ESU/H, and, lastly, we provide its time complexity analysis.

\subsection{M-ESU/H - OVERVIEW}\label{sec_m_esu_h_overview}
M-ESU/H is designed around the following three key strategies: (i) budget, (ii) server placement and upgrade, and (iii) task offloading.

\textbf{(i) Budget strategy.} \label{Sec_Sol_mesu_overiew_budget}
This strategy adopts two main policies. Firstly, for each stage $t < T$, M-ESU/H deploys a new server or upgrades the capacity of some existing servers \textit{only} when the server deployment or upgrade can increase the number of satisfied tasks.  
This \textit{as necessary} policy aims to save money for future stages when server deployment and upgrade become cheaper.
Secondly, for the last stage, i.e., $t = T$, M-ESU/H incorporates \textit{task demand prediction} that anticipates demands at stage $T+h$, their task sizes, and delay requirements, for a \textit{future time horizon} value of $h \geq 0$ (in stages).  
The forward-looking planning policy ensures greater robustness against future task profiles.

\textbf{(ii) Server placement and upgrade strategy.} 
The server placement strategy at each stage $t$ determines which nodes should host \textit{new} servers and how much capacity each server should have, while the upgrade strategy decides which \textit{existing} servers should be added with how many additional rpacks.
The strategy considers a \textit{ratio} between the number of \textit{gains} of satisfied tasks and the \textit{cost} of achieving the gain, subject to the available budget and the maximum capacity of a server at each stage $t$. 
The strategy aims to offload a \textit{cluster} of tasks to a node starting from the cluster with the highest gain-to-cost ratio.  
Detail of this strategy is shown in Algorithm~ \ref{alg2_deploy_upgrade_servers_offload_tasks}. 

\textbf{(iii) Task offloading strategy.} 
Due to the limited capacity of edge servers, M-ESU/H applies a method, called \textit{minimizing edge server utilization}, which comprises the following two main steps:
1) offload any task $r_k$ to the cloud if it can be satisfied; 
2) for each remaining task $r_k$ at stage $t$, offload the maximum size of task fraction $r^t_{k,c}$ to the cloud $c$, denoted by $b^t_{u,k}$,  so that its completion delay $\tau^t_{k,c}$ is not longer than the delay requirement of the task $r^t_k$, i.e., $\tau^t_{k,c} \leq \tau_{k,max} \times \sigma^t_k$.
The remaining fraction of task $r_k$ is then distributed between multiple edge servers so that task $r_k$ is satisfied. 
Each task in step 2) that cannot be satisfied will be completely offloaded to the cloud.
This strategy aims to use the edge servers to serve only tasks that can be satisfied. The experimental results in Section \ref{sec_evaluation}   demonstrate the benefit of the strategy.

The following proposition computes the maximum size $b^t_{u,k}$ of the task fraction $r^t_{k,c}$ used in step 2). 
\begin{prop}
For each stage $t$, the maximum size of a fraction of task $r_k^t$ that can be offloaded to cloud $c$ so that its completion delay is not longer than the delay requirement $\tau^t_{k,c}$ of the task is given as
\begin{equation}\label{eq_max_fraction}
b^t_{u,k} \leq \frac{(\tau_{k,max} \times \sigma_k)-\frac{d_{u,v}}{\pi}}{(1+\zeta) \cdot \sum_{(i,j)\in\mathcal{P}} \frac{1}{\eta_{(i,j)}} + \frac{1}{\beta_c}}.
\end{equation}
\end{prop}

\begin{proof}
To ensure that task $r^t_k$ is satisfied, the completion delay $\tau^t_{k,c}$ should be calculated from (i) the maximum time to transmit the fraction $r^t_{k,c}$ from node $u$ to cloud $c$, (ii) the maximum time to process the fraction $r^t_{k,c}$, and (iii) the maximum time to transmit the output of the fraction from cloud $c$ to node $u$.

The processing time in (ii) is calculated using Eq. (7) as $\theta_{k,c}^t = b^t_{u,k} ~/~\beta_c$. 
Let $\hat{\delta}^{''}_{u,c}$ be the sum of the maximum time in (i) and (iii). 
Calculating the times in (i) and (iii)
using Eq. (10) and Eq. (12) for distance $d_{u,c}$ and fraction size $b^t_{u,k}$, respectively, we have  
\begin{equation}{\label{eq_trans_delay_pp}}
\hat{\delta}^{''}_{u,c} = \frac{(1+\zeta) \cdot b^t_{u,k}}{\sum_{(i,j)\in\mathcal{P}_{u,c}} \eta_{(i,j)}} + \frac{d_{u,v}}{\pi}. 
\end{equation}

Consequently, to ensure that the task $r^t_{k,c}$ is satisfied, its completion time must not exceed the task delay constraint, i.e., 
$\theta_{k,c}^t + \hat{\delta}^{''}_{u,c} \leq \tau_{k,max} \times \sigma_k$. 
Thus, using (\ref{eq_trans_delay_pp}), we have $b^t_{u,k} \leq \frac{(\tau_{k,max} \times \sigma_k)-\frac{d_{u,v}}{\pi}}{(1+\zeta) \cdot \sum_{(i,j)\in\mathcal{P}} \frac{1}{\eta_{(i,j)}} + \frac{1}{\beta_c}}.$
\end{proof}
\begin{algorithm}[hbt!]
\caption{\textbf{M-ESU/H }}
\label{alg1_mesu}
\algorithmicrequire~$T$,$B$, $G^0(\mathcal{V},\mathcal{E},\mathcal{S})$, $\{\kappa_n^t\}_{t=1}^T$, $\{\kappa_u^t\}_{t=1}^T$, $\{B^t\}_{t=1}^T$, $\mu,\alpha,\rho$, $\mathcal{R}^1$, $\phi$.\\
\algorithmicensure~$\bar{\Gamma}$.
{\fontsize{10}{10}\selectfont
\begin{algorithmic}[1]
    \State $\tilde{\mathcal{R}}^1 \gets \mathcal{R}^1$,~~$t \gets 1$,~~$SP \gets \Call{\textbf{GSP}}{G^0}$
    \While{$t \leq T$}
        \State $\{\tilde{\mathcal{R}}^t, \Gamma^t\} \gets \Call{\textbf{OTC}}{\tilde{\mathcal{R}}^t, SP}$
        \State $\tilde{\mathcal{R}}^t \gets \Call{\textbf{RTS}}{\tilde{\mathcal{R}}^t, SP}$
        \If{$|\tilde{\mathcal{R}}^t| > 0~\mathbf{and}~\big(B^t \ge \kappa_n^t~\mathbf{or}~B^t \ge \kappa_u^t\big)$}
            \State $\{\Delta B^t,\tilde{\mathcal{R}}^t,\mathcal{S}^t,\Gamma^t\} \gets \Call{\textbf{DUOT}}{ B^t,\tilde{\mathcal{R}}^t,\mathcal{S}^t,SP}$
        \EndIf
        \If{$|\tilde{\mathcal{R}}^t| > 0$}
            \State $\{\tilde{\mathcal{R}}^t,\mathcal{S}^t,\Gamma^t\} \gets \Call{\textbf{OTF}}{\tilde{\mathcal{R}}^t,\mathcal{S}^t,SP}$
        \EndIf
        \If{$|\tilde{\mathcal{R}}^t| > 0$}
            \State $\{\tilde{\mathcal{R}}^t, \Gamma^t\} \gets \Call{\textbf{OTC}}{\tilde{\mathcal{R}}^t, SP}$
        \EndIf
        \State $\hat{\mathcal{R}}^{t+1} \gets \Call{\textbf{GRT}}{|\mathcal{R}^t|,\mu,\alpha,\rho}$ 
        \State $B^{t+1} \gets B^{t+1} + \Delta B^t$ 
        \State $t \gets t + 1$
    \EndWhile
    \State $\bar{\Gamma} \gets \frac{1}{T}\sum_{t=1}^{T} \Gamma^t$
\end{algorithmic}}
\end{algorithm}
\subsection{DETAILS OF M-ESU/H} \label{sec_detail_m_esu_h}
This section presents the pseudocode of (i) M-ESU/H,  (ii) function \textbf{DeployUpgradeOffloadTasks()}, and (iii) function \textbf{OffloadTaskFractions()}. The functions in (ii) and (iii) are used in M-ESU/H. 

\textbf{(i) Algorithm 1: M-ESU/H.} 
As shown in Algorithm~\ref{alg1_mesu}, Line~1 initializes stage $t=1$ and set of unsatisfied tasks $\tilde{\mathcal{R}}^1=\mathcal{R}^1$, i.e., all tasks are initially unsatisfied. In addition,  
function \textbf{GenerateShortestPaths() or \textbf{GSP()}} is used to build the table $SP$, which stores, for each node $v \in \mathcal{V}$, the shortest path computed via Dijkstra’s algorithm~\cite{buzachis2018innovative}.

In every stage~$t$, the function \textbf{OffloadTasksToCloud() or OTC()} is called at Line~3, where it attempts to offload the maximum number of tasks to the cloud subject to their delay constraints.
The function outputs the updated set of unsatisfied tasks $\tilde{\mathcal{R}}^t$ together with the number of satisfied tasks $\Gamma^t$.
Line 4 then executes function \textbf{ReduceTaskSize()} or \textbf{RTS()} to determine the maximum task fractions that can be sent to the cloud without violating deadlines, thereby reducing the workload remaining for the edge servers.
If there are some unsatisfied tasks and remaining budget to deploy new servers and/or upgrade existing servers,
Lines 5-7 use function \textbf{DeployUpgradeOffloadTasks()} or \textbf{DUOT()} that determines where to deploy or upgrade servers so that as many remaining tasks as possible become satisfied, subject to the available budget $B^t$.
This function returns the updated values of $\Delta B^t$, $\tilde{\mathcal{R}}^t$, the set of deployed servers $\mathcal{S}^t$, and the number of satisfied tasks $\Gamma^t$. 
Details of the function will be described later in Algorithm~\ref{alg2_deploy_upgrade_servers_offload_tasks}.
If there are remaining unsatisfied tasks, 
Lines 8-10 invoke function \textbf{OTF()} to partition each unsatisfied task into smaller fractions and offload them to any edge server that has residual capacity and meets the delay constraint of the task. 
Details of the function will be described later in Algorithm~\ref{alg3_offloadfraction}.
The function updates $\tilde{\mathcal{R}}^t$, $\mathcal{S}^t$, and $\Gamma^t$. 
Any remaining unsatisfied tasks will be offloaded to the cloud in Line 12 by another call to \textbf{OTC()}, 
which will return $\tilde{\mathcal{R}}^t = \{ \}$ and $\Gamma^t$ with unchanged value because none of the offloaded tasks to the cloud are satisfied.
Lines 14–16 are used to reinitialize all variables used in the next iteration of the \textbf{while} loop in Line 2, i.e., for stage $(t+1)$.
More specifically, for each stage $t < T-1$,
Line 14 uses function \textbf{GenerateRandomTasks()} or \textbf{GRT()} to generate a set of tasks at stage $t+1$ based on growth parameters $\mu, \alpha, \rho$ using (3), (4), and (5), respectively. For stage $t = T - 1$, the function generates a set of tasks at stage $T$ based on an anticipated task demands at stage $T+h$; 
recall the \textit{task demand prediction} policy of the budget strategy discussed in Section \ref{Sec_Sol_mesu_overiew_budget}.
For this case, the function also uses (3), (4) and (5) but it replaces the $"t - 1"$ exponent of each equation with $"T - 1 +h"$ to reflect the anticipated task demands at stage $"T + h"$. 
Thus, for this case, we have 
$|\mathcal{R}^{T}|=(1+\mu)^{T-1+h}\times|\mathcal{R}^{1}|$ for the number of predictive tasks, 
$\rho_b$ percent of the tasks in $\mathcal{R}^T$
will have their size increased to $b^{T}_k=(1+\alpha_b)^{T-1+h}\times b^{1}_k$, 
and $\rho_\tau$ percent of $|\mathcal{R}^T|$ tasks will have tighter deadlines given by $\tau^{T}_{\max,k}=(1-\alpha_\tau)^{T-1+h}\times \tau^{1}_{\max,k}$.
Line 15 calculates budget $B^{t+1}$ for the next stage, 
and Line 16 updates the stage to $t+1$.
Finally, Line 18 calculates the average number of tasks satisfied in the $T$ stages using Eq.~(\ref{eq_milp_average_satisfied_task}), resulting in the overall performance measure $\bar{\Gamma}$.

\textbf{(ii) Algorithm 2: Deploy, Upgrade and Offload Tasks.} 
Algorithm~\ref{alg2_deploy_upgrade_servers_offload_tasks} shows the outlines of function \textbf{DUOT()}  
which determines when and where to deploy new servers and upgrade existing ones with costs no larger than budget~$B^t$.  
Line~1 initializes the set $\mathcal{V}^i$ that stores 
all nodes to be considered for server deployment, upgrade and/or task offloadings.
The \textbf{while} loop in Line~2 is repeated if there are some unsatisfied tasks, a sufficient budget to deploy or upgrade servers, and nodes to be considered.

Line~3 uses function \textbf{Generate Tasks Cluster ()} or \textbf{GTC()} to generate a set $\mathcal{R}^c$ of \textit{candidate} task clusters  from the remaining unsatisfied tasks in $\tilde{\mathcal{R}}^t$.
More specifically, for each node $s \in \mathcal{V}^i$, function \textbf{GTC()} performs the following two steps: 
1) Generate a \textit{candidate} cluster of tasks $\mathcal{R}^c_s \in \mathcal{R}^c$, from set of unsatisfied tasks $\tilde{\mathcal{R}}^t$, that can be offloaded to node $s$ so that they become satisfied, subject to the maximum capacity $\psi_s$ of the server and 
the available budget $B^t$. 
The function includes two options: (i) if the node does not have a server, it considers deploying a new server with $m_s^t$ rpacks such that $\mathcal{R}^c_s$ has the maximum number of tasks, denoted by $\Delta\Gamma^{(n)}_s$, or (ii) if the node has a server, it considers upgrading the server with $m_s^t \leq M_{max} - M_s^{t-1}$ rpacks such that $\mathcal{R}^c_s$ has the maximum number of tasks, denoted by $\Delta\Gamma^{(u)}_s$. To maximize the number of tasks in $\mathcal{R}^c_s$, \textbf{GTC()} considers tasks in $\tilde{\mathcal{R}}^t$ to be included into $\mathcal{R}^c_s$ starting from tasks with the smallest size;  
2) Calculate the value of the gain-to-cost ratio for each cluster $\mathcal{R}^c_s$. The cost for option (i), denoted by  
$\kappa_n^t$, is calculated using Eq. (\ref{server_deployment_cost}), while the cost for option (ii) is $m_s^t \times \kappa_u^t$. Thus, the gain to cost ratio for option (i) and (ii) are $\frac{\Delta\Gamma^{(n)}_s}{\kappa_n^t}$ and $ \frac{\Delta\Gamma^{(u)}_s}{m_s^t \times \kappa_u^t}$, respectively. 

If in Line 4 there are some candidate clusters, Line~5 uses function \textbf{GetFirst()} to select a cluster $\mathcal{R}_s^c$ with the largest gain-to-cost ratio.
If~$s$ does not already host a server, i.e., option (i) of \textbf{GTC()}, Line~7 deploys a new server with an initial $m_s^t$ rpacks, 
and Line 8 subtracts its deployment cost~$\kappa_n^t$ from available budget from~$B^t$.  
Otherwise, Line 10 increases the number of rpacks that have been installed on the server $s$ up to stage $t$, and Line 11 updates the residual budget. 
Line~13 then uses function \textbf{Offload Task()} or \textbf{OT()}
to offload all tasks in cluster ~$\mathcal{R}_s^c$ to server~$s$.
This function updates the remaining unsatisfied task set~$\tilde{\mathcal{R}}^t$, the server set~$\mathcal{S}^t$, 
and the satisfied-task counter~$\Gamma^t$. 
Line~15 removes~$s$ from future consideration. 
Finally, Line~17 stores the remaining budget as~$\Delta B^t = B^t$, which will be carried forward to the next stage $t+1$, and Line~18 returns the updated values of~$\Delta B^t$, $\tilde{\mathcal{R}}^t$, $\mathcal{S}^t$, and~$\Gamma^t$.

\textbf{(iii) Algorithm 3: Offload Task Fractions.} 
The function \textbf{OffloadTaskFractions()} or \textbf{OTF()}, presented in Algorithm~\ref{alg3_offloadfraction}, aims to maximize the number of satisfied tasks by partitioning each unsatisfied task into multiple fractions and offloading these fractions to \textit{multiple} edge servers that have available capacity.
Unlike \textbf{DUOT()}, which offloads tasks without considering fractional allocation, \textbf{OTF()} enables finer-grained resource utilization through task splitting.
In Line~1, the function \textbf{LeftOverCapacityServers()} or \textbf{LOCS()} identifies the set of servers $\mathcal{S}_l$ that still have unused capacity.
Lines~2–15 then iterate over each unsatisfied task $r_k^t$, attempting to divide it into fractions that can be distributed across servers in $\mathcal{S}_l$. 
For each server $s \in \mathcal{S}_l$ with remaining capacity in Lines 3-4, Line~5 computes the largest possible fraction of $r_k^t$ that can be offloaded to $s$.
Next, function \textbf{FractionDelay()} or \textbf{FD()} in Line~6 evaluates the delay of fraction $r_{k,s}$ when executed on $s$.
If the delay requirement in Line 6 is met, 
Line~7 updates the size of task $r_k^t$ that has not been offloaded,  
and Line~8 includes this fraction to the set of satisfied fractions for task $r_k^t$.
If all fractions of task $r_k^t$ are satisfied,  Line~13 uses function \textbf{OffloadFractionsToServers()} or \textbf{OFS()} to offload all satisfied fractions in $\mathcal{F}_k^t$ to their corresponding servers, 
update the remaining server capacities and unsatisfied tasks, and increment the count of satisfied tasks.
Finally, Line~16 returns remaining unsatisfied tasks $\tilde{\mathcal{R}^t}$, the updated set of deployed servers 
$\mathcal{S}^t$ and the total number of satisfied tasks $\Gamma^t$ at stage~$t$.

\begin{algorithm}
\caption{Deploy-Upgrade and Offload Tasks (DUOT).} \label{alg2_deploy_upgrade_servers_offload_tasks}
\algorithmicrequire \ ${B}^t,\mathcal{V},\tilde{\mathcal{R}}^{t}, \mathcal{R}^T, \mathcal{S}^t, SP, t, T$.
\algorithmicensure~ $\Delta{B}^t$, $\tilde{\mathcal{R}}^t$, $\mathcal{S}^t, \Gamma^t$.
\begin{algorithmic}[1]
    \State $\mathcal{V}^i \gets \mathcal{V}$
    \While {$|\tilde{\mathcal{R}}^t|>0$~\textbf{and}~$B^t \geq$\textbf{min}($\kappa_n^t,\kappa_u^t)$~\textbf{and}~$|\mathcal{V}^i|>0$}
        \State \{$\mathcal{R}^c\} \gets \Call{\textbf{GTC}}{\tilde{\mathcal{R}}^t, \mathcal{V}^i, B^t, SP}$
        \If {$|\mathcal{R}{}^c|>0$}
            \State {$s \gets \Call{\textbf{GetFirst}}{\mathcal{R}^c}$}
            \If {$s \notin \mathcal{S}^t$}
                \State $\mathcal{S}^t \gets \mathcal{S}^t + 
                \{s\}$
                \State $B^t \gets B^t - \kappa^t_n$  
            \Else
                \State $M^t_s \gets M^t_s + m^t_s$ 
                \State $B^t \gets B^t - (\kappa^t_u \times m^t_s)$
            \EndIf
            \State \{$\tilde{\mathcal{R}}^t, \mathcal{S}^t ,\Gamma^t \} \gets$ \Call{\textbf{OT}}{$\mathcal{R}^c_s$}         
        \EndIf
        \State {$\mathcal{V}^i \gets \mathcal{V}^i - s$}
    \EndWhile
    \State $\Delta{B}^t \gets {B}^t$
    \State \textbf{return} $\{\Delta{B}^t, \tilde{\mathcal{R}}^t,\mathcal{S}^t,\Gamma^t\}$
\end{algorithmic}
\end{algorithm}

\begin{algorithm}
\caption{Offload Task Fractions (OTF).}
\label{alg3_offloadfraction}
\algorithmicrequire~$\tilde{\mathcal{R}}^t,  \mathcal{S}^t$.
\algorithmicensure~$\tilde{\mathcal{R}}^t, \mathcal{S}^t, \Gamma^t$.
\begin{algorithmic}[1]
    \State $\mathcal{S}_l \gets$ 
    \Call{\textbf{LOCS}}{$\mathcal{S}^t$}
    \For{$r_k^t \in \hat{\mathcal{R}^t}$}
        \For{$s \in \mathcal{S}_l$}
            \If{$b_k^t>0$ \textbf{and} 
            $\psi_s^t > 0$}
                \State $b_{k,s}^t \gets \Call{\textbf{min}}{b_k^t, \psi_s^t}$
                \If{\textbf{FD}$(b_{k,s}^t, s) \leq \tau_{k,max}^t$}
                \State $b_k^t \gets b_k^t - b_{k,s}^t$
                \State $\mathcal{F}_k^t \gets \mathcal{F}_k^t + \{r_{k,s}^t\}$
                \EndIf
            \EndIf
        \EndFor
    \If{$b_k^t = 0$}
         \State $\{\Tilde{\mathcal{R}}^t, \mathcal{S}^t, \Gamma^t \} \gets$ 
     \Call{\textbf{OFS}}{$\mathcal{F}_k^t$}
    \EndIf
    \EndFor
    \State \textbf{return} $\{\tilde{\mathcal{R}}^t,\mathcal{S}^t, \Gamma^t\}$
\end{algorithmic}
\end{algorithm}
 
To conclude this section, we present the time complexity of M-ESU/H as follows.
%
\begin{prop}\label{sec:timeComplexity}
The time complexity of M-ESU/H is $O(|\mathcal{R}^t||\mathcal{V}|^3T)$. 
\end{prop}
\begin{proof}\label{Proof:TimeComplexity}
At Line 1, function \textbf{GSP()} applies Dijkstra’s algorithm to generate 
a table $SP$ that contains the shortest-path for all nodes $v \in \mathcal{V}$. 
The function requires
$O(|\mathcal{V}|^2 \log |\mathcal{V}| + |\mathcal{V}||\mathcal{E}|)$ time.  
Function \textbf{OTC()} at Line 3 requires $O(|\mathcal{R}^t||\mathcal{V}|)$ because calculating delay for each task takes $O(|\mathcal{V}|)$.
At Line~4, function \textbf{RTS()} requires $O(|\mathcal{R}^t||\mathcal{V}|)$ steps because it considers each task and computes the delay of one fraction of the task. 
At Line~5, function \textbf{DUOT()} is dominated by the time to evaluate candidate nodes and servers for deployment.
Specifically, function \textbf{GTC()} takes $O(|\mathcal{R}^t|\,|\mathcal{V}|^{2})$,
function \textbf{Getfirst()} requires $O(1)$ and function \textbf{OT()} takes $O(|\mathcal{R}^t|)$;
these are repeated at most $O(|\mathcal{V}|)$ times. 
Thus, function \textbf{DUOT()} requires $O(\mathcal{R}^t|\,|\mathcal{V}|^3)$.
At Line~9, function \textbf{OTF()}, shown in Algorithm~\ref{alg3_offloadfraction},  handles the fractional offloading process. 
Function \textbf{LOCS()} requires $O(|\mathcal{S}^t|)$.
Function \textbf{FD()} executes in $O(|\mathcal{V}|)$ times to get fraction delay, and it is repeated $O(|\mathcal{R}^t||\mathcal{S}^t|)$ times.

Function \textbf{OFS()} has $O(|\mathcal{S}^t|)$ complexity, and it is repeated   $O(|\mathcal{R}^t|)$ times, resulting in a total complexity of $O(|\mathcal{S}^t|+|\mathcal{R}^t||\mathcal{S}^t||\mathcal{V}| + |\mathcal{R}^t||\mathcal{S}^t|)$ = $O(|\mathcal{R}^t||\mathcal{S}^t||\mathcal{V}|)$.
At Line~12, \textbf{OTC()} is called again to offload any remaining tasks to the cloud, contributing $O(\mathcal{R}^t|\,|\mathcal{V}|)$. 
At Line~14, function \textbf{GRT()} generates new tasks for the next stage, requiring $O(|\mathcal{R}^t| + |\mathcal{V}| \log |\mathcal{V}|)$ to update the task set and recalculate delay attributes.
Function \textbf{OTC()}, \textbf{RTS()}, \textbf{DUOT()}, \textbf{OTF()},  and \textbf{GRT()} are repeated $T$ times, so the time complexity of M-ESU/H is
$O(|\mathcal{V}|^2\log|\mathcal{V}|+
|\mathcal{V}||\mathcal{E}|+3|\mathcal{R}^t||\mathcal{V}|T+|\mathcal{R}^t||\mathcal{V}|^3T+|\mathcal{R}^t||\mathcal{S}^t||\mathcal{V}|T+|\mathcal{R}^t|T+|\mathcal{V}|\log|\mathcal{V}|T)$
= $O(|\mathcal{R}^t||\mathcal{V}|^3T)$
because in general we have $|\mathcal{R}^t| > |\mathcal{E}| > |\mathcal{V}|$.
\end{proof}

\section{EVALUATION}\label{sec_evaluation}
We compare the performance of M-ESU/H (or /H) against MILP and three versions of M-ESU, namely M-ESU/DO (or /DO), M-ESU/UF (or /UF), and M-ESU/DF (or /DF). 
Briefly, /DO is the M-ESU that only considers server deployment; it does not consider capacity upgrades on any existing server. 
On the other hand, /UF and /DF consider both server deployment and server upgrade. More specifically, /UF aims at increasing the capacity of the existing server first. It deploys new servers only after each existing server reaches its maximum capacity or when there is no existing server that can produce satisfied tasks.  
In contrast, in each stage, /DF considers a new server deployment first. Then, if the residual budget cannot be used to deploy an additional server or there is no available location to deploy a new server, it will use the remaining budget to increase the capacity of the deployed servers. 
All algorithms are executed on a 64-bit Windows system equipped with a Snapdragon X Elite 3.4 GHz processor and 32 GB of memory. Each reported value represents the average over ten runs.
The following describes the evaluation setup for (i) networks, (ii) servers, (iii) tasks,  (iv) budget, and (v) cost. Table~\ref{tbl_simulation_setting} lists the parameters used in our experiments.
\begin{table}[htbp]
    \centering
    \caption{Simulation settings.}
    \label{tbl_simulation_setting}
    \setlength{\tabcolsep}{2pt}
{\fontsize{8}{10}\selectfont
\begin{tabular}{|cc|cc|}
\hline
\multicolumn{1}{|c|}{\textbf{Parameters}}         & \textbf{Value}          & \multicolumn{1}{c|}{\textbf{Parameters}}   & \textbf{Value} \\ \hline
\multicolumn{2}{|c|}{\textbf{Networks}}                                     & \multicolumn{2}{c|}{\textbf{Servers}}                       \\ \hline
\multicolumn{1}{|c|}{$\zeta$}                     & 0.1                     & \multicolumn{1}{c|}{$\beta_s$ ($\beta_c$)} & 10 (10)        \\
\multicolumn{1}{|c|}{$\eta_{u,v}$}                & \{20,40\} Gb/s          & \multicolumn{1}{c|}{$\psi_s$}              & 10-40 Gb/s     \\
\multicolumn{1}{|c|}{$\eta_{u,c}$}                & \{2,5\} Gb/s            & \multicolumn{1}{c|}{$\psi_c$}              & $\infty$      \\
\multicolumn{1}{|c|}{$\delta^{'}_{u,v}$}          & 0 ms                    & \multicolumn{1}{c|}{$m_{min}$}             & 1              \\
\multicolumn{1}{|c|}{$\delta^{'}_{u,c}$}          & 50 ms                   & \multicolumn{1}{c|}{$M_{max}$}             & 4              \\ \hline
\multicolumn{2}{|c|}{\textbf{Tasks}}                                        & \multicolumn{2}{c|}{\textbf{Budget}}                        \\ \hline
\multicolumn{1}{|c|}{$|\mathcal{R}^1|$}           & $3 \times|\mathcal{V}|$ & \multicolumn{1}{c|}{Pc}                    & 600            \\
\multicolumn{1}{|c|}{$\mu$}                       & 50\%                    & \multicolumn{1}{c|}{$\kappa^1_u$}          & 100            \\
\multicolumn{1}{|c|}{$\rho_\tau$ ($\alpha_\tau$)} & 20(50)\%                & \multicolumn{1}{c|}{$\phi$}                & 0.5            \\
\multicolumn{1}{|c|}{$\rho_b$ ($\alpha_b$)}       & 20(50)\%                & \multicolumn{1}{c|}{$B(x\%)$}              & 75\%           \\ \hline
\end{tabular}}
\end{table}

(i) \textit{Networks:} 
We employ the three synthetic MEC networks provided in \cite{xiang2021dataset}: 25N50E, 50N50E, and 100N150E. The notation ''$x$\textbf{N}$y$\textbf{E}'' specifies a network with $x$ nodes 
and $y$ links; for instance, 100N150E corresponds to $|\mathcal{V}| = 100$ nodes and $|\mathcal{E}| = 150$ links.
We additionally construct four smaller synthetic networks—5N5E, 10N23E, 15N27E, and 20N48E because MILP was unable to produce results for the three larger networks after 48 hours of execution.
We consider each AP and the cloud $c$ is connected by fiber or copper links \cite{xiang2021dataset}.
We set the transmission rate from any server $u$ to the cloud $c$, i.e.,  $\eta_{u,c}$, to a value between two and five Gb/s.  
For the transmission rate $\eta_{u,v}$ between any pair of edge servers $u$ and $v$, we set it to a value from the set \{20, 40\} (in Gb/s). 
Note that these values follow the study in \cite{ranaweera2020novel}, which stated that 
tasks may require high transmission rates on the order of tens of gigabits per second.
In line with the setting used in~\cite{wang2017controller}, we assign a propagation delay of 50 ms between each AP and the cloud.

(ii) \textit{Servers:}  
Following the work in \cite{xiang2021dataset}, 
we set the capacity $\psi_s$ of each edge server $s$ to an equal value between 10 and 40 (in Gbyte).
We assume unlimited capacity for the cloud server. Based on~\cite{yang2019cloudlet}, the processing speeds both for edge servers and cloud are assigned as $\beta_s = 10$ and $\beta_c = 10$, respectively, in units of Giga cycles/s.
All networks initially have servers arbitrarily deployed on 50\% of the nodes, each with 50\% 
capacity (i.e., two rpacks installed, given $M_{max}=4$). 
Each new server deployment in /DO has two rpacks installed, whereas the other heuristic methods allow flexible deployment with capacities ranging from one to four rpacks.

(iii) \textit{Tasks:}
In stage $t = 1$, we arbitrarily distribute 100 tasks between all APs.
Following \cite{ren2021demand}, the number of tasks increases by $\mu = 50\%$ per stage.
Thus, in our simulation, we generate $\mu \times |\mathcal{R}^t|$ new tasks at each stage $t+1$, for an increase rate $\mu = 50\%$. 
We arbitrarily set each task $r_k$ as intolerant or tolerant by setting its delay tolerant $\sigma$ to a value equal to one or larger than one, respectively. 
We set $\rho_b$ and $\alpha_b$ to 20\% and 50\%, respectively. This means 20\% of tasks will increase their size by 50\%. To simulate the stringent deadline, we set $\rho_\tau$ and $\alpha_\tau$ to 20\% and 50\%, respectively, meaning that 20\% of tasks will decrease their deadline by 50\%.
We then perform a random assignment of tasks to nodes $u \in \{\mathcal{V} - c\}$.
Task deadlines and sizes are configured to reflect realistic MEC workloads in which edge servers operate under limited capacity and cannot satisfy all incoming tasks, while cloud offloading continues to support a small portion of tasks as reported in \cite{ranaweera2020novel}.
Thus, the maximum task deadline $\tau_{max,k}$ is arbitrarily selected from the set $\{3, 5, 10\}$ (in second), 
while the task size $b_k^1$ is randomly chosen from $\{10, 20, 30\}$ (in Gb).
We consider each task to produce a result that is 10\% of its size and, thus, we have $\zeta=0.1$. 

(iv) \textit{Budget and cost: }
In the initial stage $t=1$, the cost $I^t$ of deploying the required infrastructure is $I^1 = 600$, while the cost $\kappa_u^t$ of each rpack is $\kappa_u^1 = 100$. 
Both costs decrease at a rate of $\phi = 0.2$ per stage.
We set the maximum number of rpacks that can be installed on a server to $M_{max}=4$.  
For a given network $G(\mathcal{V}, \mathcal{E})$, we define $B(x\%)$ as the budget required to install an edge server equipped with the maximum number of rpacks $M_{\max}$ at $x\%$ of the nodes in $|\mathcal{V}|$.
The $B(x\%)$ considers deployment of edge server at stage $t=1$, i.e., cost $\kappa_n^1$,  
that includes $M_{max}$ rpacks, i.e., the server is at full capacity.  
Consider the case $x = 75\%$. Under this setting, the 25N50E network requires a budget of $B(75\%) = 0.75 \times 25 \times 1000 = 18750$, where $\kappa_n^1 = 600 + 4 \times 100 = 1000$.
Likewise, for the 50N50E and 100N150E networks, the corresponding budgets at $B(75\%)$ are 37500 and 75000, respectively.

\subsection{SMALL NETWORKS}\label{sec_small_networks}
This section compares the performance of /H, /DO, /DF, and /UF in terms of their percentage of satisfied tasks $\bar{\Gamma}$, and their running time. 
The evaluation uses four small-scale network topologies: 5N5E, 10N23E, 15N27E, and 20N48E. Other parameters are shown in Table \ref{tbl_effectiveness}.
Each topology $G(\mathcal{V}, \mathcal{E}, \mathcal{S})$ considers four sets of tasks at stage $t=1$. The number of tasks for each set is calculated from the number of nodes, i.e.,  $|\mathcal{R}^1|$ is set to $1\times|\mathcal{V}|$, $3\times|\mathcal{V}|$, $5\times|\mathcal{V}|$, and $7\times|\mathcal{V}|$, respectively.
The evaluation uses a total budget $B(75\%)$ and the number of stages $T = 3$. 

As shown in Table \ref{tbl_effectiveness}, MILP produces the highest percentage of satisfied tasks among the other four methods. However, it was unable to produce results, denoted by "N/A", for some of the experiments, after running for 48 hours.
In contrast, the four heuristic methods can produce results in less than 1 second for each configuration. 
For example, on network $20$\textbf{N}$48$\textbf{E} with $|\mathcal{R}^1|=7 \times 20 = 140$ tasks, /H, /DF, /UF, and /DO produce results in 0.16, 0.52, 0.46, and 0.02 seconds, respectively. 
We exclude the details of their running time in Table \ref{tbl_effectiveness} for conciseness. 
\begin{table}[htbp]
    \centering
    \caption{Percentage of satisfied tasks.}
    \label{tbl_effectiveness}
    \setlength{\tabcolsep}{3pt}
{\fontsize{8}{10}\selectfont
\begin{tabular}{ccccccc}
\hline
\multirow{2}{*}{\textbf{Topology}} &
  \multirow{2}{*}{\boldmath$|\mathcal{R}^1|$} &
  \multirow{2}{*}{\textbf{MILP}} &
  \multicolumn{4}{c}{\textbf{M-ESU}} \\ \cline{4-7} 
       &     &       & \textbf{/H} & \textbf{/DF} & \textbf{/UF}              & \textbf{/DO} \\ \hline
5N5E   & 5   & 100   & 99.55       & 99.09        & 98.64                     & 95.45        \\
5N5E   & 15  & 98.57 & 96.29       & 88.71        & 88.57                     & 78.71        \\
5N5E   & 25  & N/A   & 88.72       & 75.64        & 76.32                     & 67.18        \\
5N5E   & 35  & N/A   & 78.91       & 66.73        & 68.12                     & 61.45        \\ \hline
10N23E & 10  & 100   & 99.79       & 99.57        & 98.30                     & 95.11        \\
10N23E & 30  & 92.96 & 92.54       & 79.37        & 80.07                     & 63.59        \\
10N23E & 50  & N/A   & 79.49       & 63.76        & 64.35                     & 51.48        \\
10N23E & 70  & N/A   & 68.61       & 54.40        & 54.85                     & 44.94        \\ \hline
15N27E & 15  & 100   & 98.43       & 97.14        & 96.71                     & 95.43        \\
15N27E & 45  & 100   & 97.50       & 89.58        & 89.34                     & 78.44        \\
15N27E & 75  & N/A   & 92.28       & 77.27        & 77.35                     & 69.80        \\
15N27E & 105 & N/A   & 85.86       & 72.60        & 73.20                     & 66.90        \\ \hline
20N48E & 20  & 100   & 99.16       & 98.00        & \multicolumn{1}{l}{98.21} & 97.16        \\
20N48E & 60  & 99.65 & 98.07       & 88.53        & \multicolumn{1}{l}{89.68} & 79.58        \\
20N48E & 100 & N/A   & 92.04       & 76.95        & \multicolumn{1}{l}{78.78} & 70.88        \\
20N48E & 140 & N/A   & 85.50       & 72.17        & \multicolumn{1}{l}{73.49} & 66.77        \\ \hline
\multicolumn{3}{c}{\textbf{Average}} &
  \multicolumn{1}{l}{\textbf{90.80}} &
  \multicolumn{1}{l}{\textbf{81.22}} &
  \multicolumn{1}{l}{\textbf{81.62}} &
  \multicolumn{1}{l}{\textbf{73.93}} \\ \hline
\end{tabular}}
\end{table}

Table \ref{tbl_effectiveness} shows that for some experiments with a small number of tasks per node, e.g., $10$\textbf{N}$23$\textbf{E} $|\mathcal{R}^1|= 10$ and $20$\textbf{N}$48$\textbf{E} with $|\mathcal{R}^1|= 30$, /H produces results only 0.21 and 0.42 percentage points lower than the 100\% of satisfied tasks achieved by MILP.   
On the other hand, the table also shows that /DO performs the worst.  These results highlight the benefits of incorporating server capacity upgrades, which enable the system to handle more tasks compared to configurations that rely solely on new deployments.
Between the two versions of M-ESU that allow server capacity upgrade, the /H version performs the best, i.e., on average achieves 9.58 and 9.18 percentage points larger $\bar{\Gamma}$ than /DF and /UF, respectively. 
The reason is because /H can flexibly select the better option between server deployment and server upgrade. On the other hand, /DF (/UF) must prioritize server deployment (server upgrade) first, even if there is a better option for server upgrade (server deployment) that would have been able to increase the percentage of satisfied tasks.

Table \ref{tbl_effectiveness} also shows that the /UF version produces a higher percentage of satisfied tasks than the /DF version. The reason is because with increasing task demand,  some existing servers can be overloaded with their local task demands. In this case, demands must be offloaded to other servers that may not be able to satisfy their delay requirement. Thus, upgrading the capacity of the servers first for this case is the better option as compared to deploying new servers first in /DF.
%
\subsection{LARGE NETWORKS}\label{sec_large_networks}
This section evaluates the effect of (1) budget, 2) number of tasks, 3) ratio between server deployment and server upgrade costs, and 4) number of stages on the percentage of satisfied tasks, namely $\bar{\Gamma}$.  
We use networks $25$\textbf{N}$50$\textbf{E}, $50$\textbf{N}$50$\textbf{E}, and $100$\textbf{N}$150$\textbf{E}, and the parameters shown in Table \ref{tbl_simulation_setting}.

\subsubsection{Budget Allocation}
The analysis considers budget allocations ranging from B(40\%) to B(100\%) over $T=3$ stages. The initial number of tasks is set to $|\mathcal{R}^1|=3|\mathcal{V}|$, and this workload increases by 50\% ($\mu=0.5$) at each subsequent stage.

 \begin{figure}[ht]
  \centering
 \includegraphics[width=0.5\textwidth, keepaspectratio]{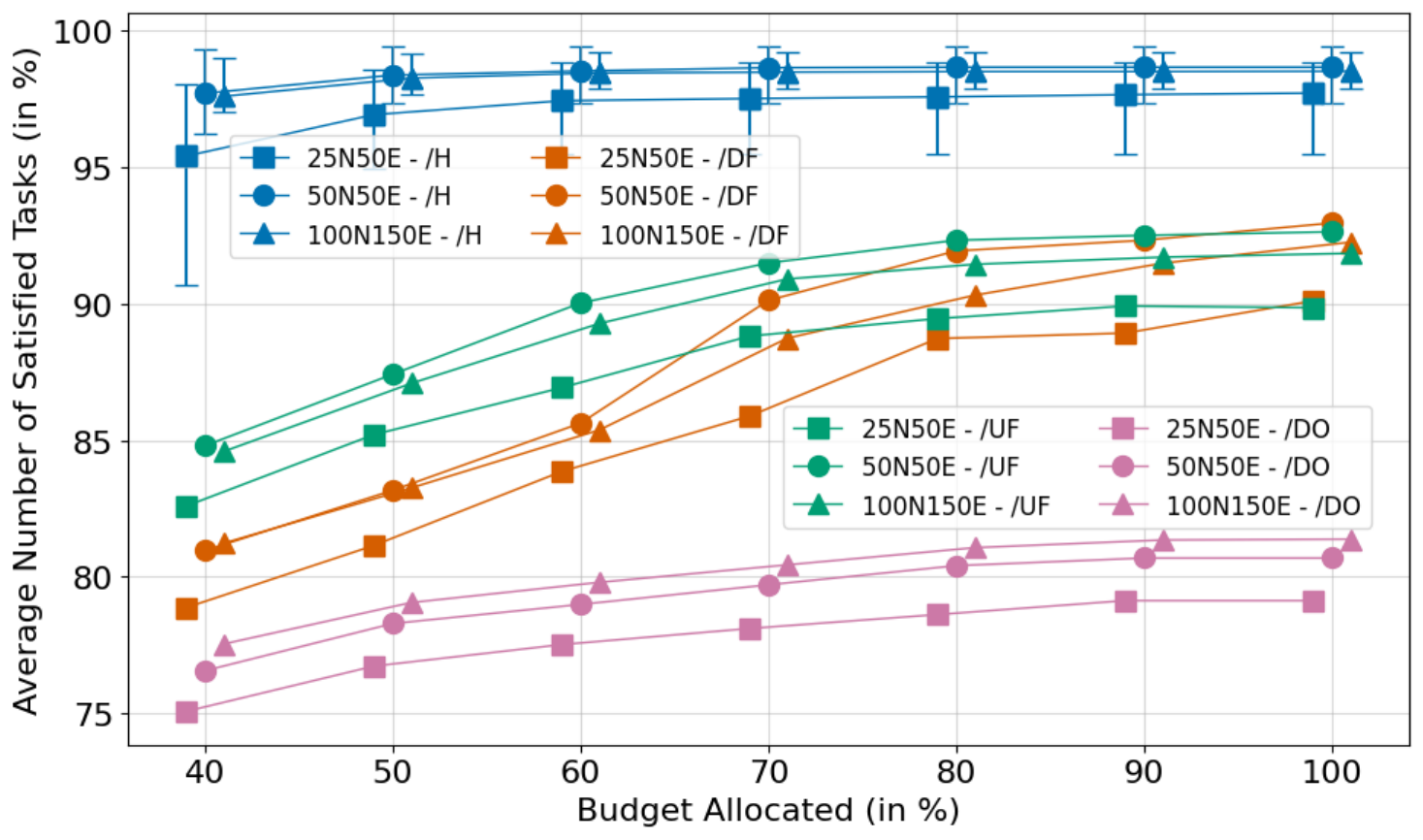}
  \caption{Amount of budget versus the percentage of satisfied tasks.}
  \label{fig_budgets}
\end{figure}

Fig. \ref{fig_budgets} shows the impact of various budgets on the percentage of $\bar{\Gamma}$. The figure also displays the corresponding minimum and maximum values 
for /H, with standard deviations ranging from 0.35 to 2.09. This small variation indicates that M-ESU/H delivers stable performance across repeated runs, with only minor fluctuations in the percentage of satisfied tasks.
In general, increasing the budget from  B(40\%) to B(100\%) improves $\bar{\Gamma}$.
The reason is because more budget allows all algorithms to increase the number of servers that can be deployed and/or upgraded.
Fig. \ref{fig_budgets} also shows that /H achieves the highest $\bar{\Gamma}$, i.e., $\bar{\Gamma} = 97.99\%$, outperforming /DF, /UF, and /DO, which achieve an average percentage of 87.02\%, 89.09\%, and 79.06\%, respectively. 
In addition, even when it uses a lower budget, the version /H produces a higher $\bar{\Gamma}$ than the other three versions, e.g., B(40\%) of /H is better than B(100\%) of /DF, /UF, and /DO. The superiority of the version /H highlights the importance of using a better server placement, upgrade, and task offload strategy to produce higher $\bar{\Gamma}$.
Consistent with the result in \ref{sec_small_networks}, Fig. \ref{fig_budgets} shows that the upgrade-first version of M-ESU is slightly better than the deploy-first version.
The worst result of /DO further highlights the benefit of being able to upgrade server capacity in the other three methods. 
 \begin{figure}[ht]
  \centering
 \includegraphics[width=0.5\textwidth, keepaspectratio]{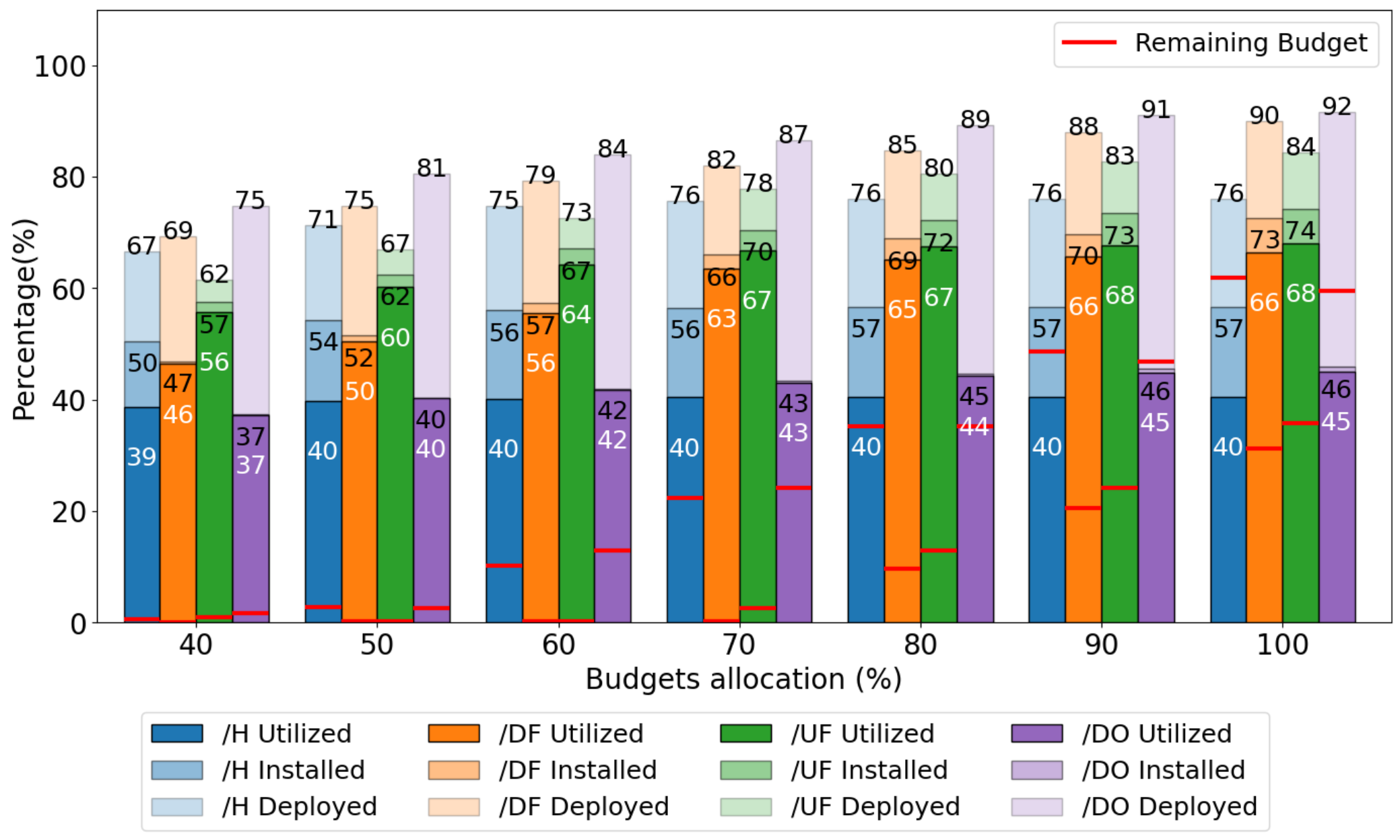}
  \caption{Budget level versus deployed servers, installed rpacks, and server utilization.}
  \label{fig_budget_servers_v2}
\end{figure} 

To see how each algorithm spends its allocated budget, we present Fig.  \ref{fig_budget_servers_v2}, where 
$\hat{\Delta} B^T$ (in \%) is the ratio between the remaining budget at stage $T$ and the total budget $B$, i.e., $\hat{\Delta} B^T = (\Delta B^T / B) \times 100\%$, 
$\hat{\mathcal{S}}^T$ (in \%) is the ratio between the total deployed servers up to stage $T$ and the number of nodes in the network, i.e., $\hat{\mathcal{S}}^T = (|\mathcal{S}^T| / |\mathcal{V}|) \times 100\%$, 
$\hat{M}^T$ (in \%) is the ratio between the total installed rpacks up to stage $T$ and the maximum possible rpacks that can be installed in the network, i.e., $\hat{M}^T = \left( \sum_{\forall s \in \mathcal{V}} M^T_s / \mathcal{M} \right) \times 100\%$, 
and $\mathcal{C}^T$ (in \%) is the server utilization, calculated as the ratio between the total workload of all servers at stage $T$ and the capacity of all servers when they have the maximum number of rpacks, i.e., $\mathcal{C}^T = (\sum_{s\in(\mathcal{S}^T-c)}\hat{\psi_s^T})/(|\mathcal{V}|\cdot\psi) \times 100\%$.

Fig. \ref{fig_budget_servers_v2} shows the percentage $\hat{\Delta} B^T$, $\hat{\mathcal{S}}^T$,  $\hat{M}^T$, and $\mathcal{C}^T$ that are produced by the three algorithms for network $50$\textbf{N}$50$\textbf{E}, i.e., $|\mathcal{V}|=50$ and $M_{max} = 4$.
As expected, Fig. \ref{fig_budget_servers_v2} shows that increasing budget from B(40\%) to B(100\%) allows each algorithm to deploy more servers, e.g., from $\hat{\mathcal{S}}^T = 75\%$ to $\hat{\mathcal{S}}^T = 92\%$ for /DO. 
The figure also shows that /H and /UF deployed a smaller number of servers compared to /DF and /DO, e.g., for B(80\%), edge servers are deployed in 76\%, 80\%, 85\%, and 89\% on the 50 possible nodes in the network, or 38, 40, $\approx43$, and $\approx45$ deployed servers, respectively. 
The reason is because /DO can use the budget only to deploy servers and /DF prioritizes server deployment.

In terms of the number of installed rpacks, Fig. \ref{fig_budget_servers_v2} shows that the value of $\hat{M}^T$ for /H, /DF, /UF, and /DO increase as the budget increases from B(40\%) to B(100\%), ranging from 50-57\%, 47–73\%, 57–74\%, and 37–46\%, respectively. 
This growth is reflected in the corresponding linear increase in the number of satisfied tasks.

In terms of server utilization $\mathcal{C}^T$, Fig. \ref{fig_budget_servers_v2} shows that, except for /H, all methods utilize nearly all of their available capacity of the network to execute tasks, with /DF, /UF, and /DO using, 
on average, 96\%, 94\%, and 99\% capacity of the installed rpacks respectively. 
In contrast, /H achieves better $\bar{\Gamma}$ while using only on average 72\% of capacity. This highlights that /H is more efficient, achieving higher task satisfaction with substantially lower resource utilization.

In terms of the remaining budget $\hat{\Delta}B^T$, Fig. \ref{fig_budget_servers_v2} shows (in red line) that all algorithms retain more unused budget as the budget allocation increases. 
Interestingly, /DO has the largest remaining budget of 26.15\% despite deploying more servers, followed by /H, /DF, and /UF with 25.97\%, 11.01\%, and 8.85\%, respectively. 
The reason is because /DO does not perform upgrades; hence, the majority of its budget is directed toward new deployments (e.g., 92\% in /DO compared to 76\%, 84\%, and 90\% in /H, /UF, and /DF, respectively). 
In contrast, /H shows a sharp increase in the remaining budget beginning at $B(60\%)$, as it no longer deploys additional servers or upgrades the server capacity, leading to significant savings as the budget increases.

\subsubsection{Number of Tasks}
We evaluate how the number of tasks influences the task satisfaction rate $\Bar{\Gamma}$ using the 25N50E, 50N50E, and 100N150E networks, under a 75\% coverage budget $B(75\%)$ and $T = 3$. The number of tasks is configured as $2|\mathcal{V}|$, $4|\mathcal{V}|$, $6|\mathcal{V}|$, and $8|\mathcal{V}|$.
Fig.~\ref{fig_tasks} illustrates how increasing the number of tasks influences the performance of /H, /DF, /UF, and /DO.
The figure also presents the minimum and maximum outcomes for /H, whose 
standard deviation lies between 0.22 and 1.92. This narrow spread 
demonstrates that the algorithm’s performance remains highly consistent despite 
variations in task demands.
As expected, satisfaction decreases for larger number of tasks.
However, /H experiences the smallest average decrease as compared to other methods. In particular,  when the number of tasks increases from $2|\mathcal{V}|$ to $8|\mathcal{V}|$, the average decrease is 14.69, 26.25, 24.89, and 22.48 percentage point for /H, /DF, /UF, and /DO, respectively. This indicates that /H is more robust to task growth. 

To further analyze performance, Fig. \ref{fig_tasks_servers2} shows the percentage of servers deployed $\hat{\mathcal{S}}^T$, installed rpacks $\hat{M}^T$, and server utilization $\mathcal{C}^T$, when the number of tasks increases from $2|\mathcal{V}|$ to $8|\mathcal{V}|$ for network 50\textbf{N}50\textbf{E}.
It can be observed that, in addition to producing the highest percentage of satisfied tasks $\bar{\Gamma}$, /H is also more efficient in utilizing computing resources by dynamically adjusting both the number of servers and the number of installed rpacks as per task demands. 
In contrast, other methods tend to maximize the number of deployed servers starting from $|\mathcal{R}|=2|\mathcal{V}|$; e.g., as the number of tasks increases from 
$|\mathcal{R}|=2|\mathcal{V}|$ to 
$|\mathcal{R}|=8|\mathcal{V}|$, the percentage of servers deployed under /H rises gradually from 69\% to 81\%. 
For comparison, the deployment ratios for /DF, /UF, and /DO  remain fixed at 84\%, 77\%, and 88\%, respectively. 
A similar pattern can be seen in $\mathcal{C}^T$, i.e., its percentage for /H increases from 47\% at $|\mathcal{R}|=2|\mathcal{V}|$ to 73\% at $|\mathcal{R}|=8|\mathcal{V}|$, while the percentage for /DF and /UF remains capped at around 64–68\% and 77\%, and that for /DO is fixed at 44\%.
 \begin{figure}[ht]
  \centering
 \includegraphics[width=0.5\textwidth, keepaspectratio]{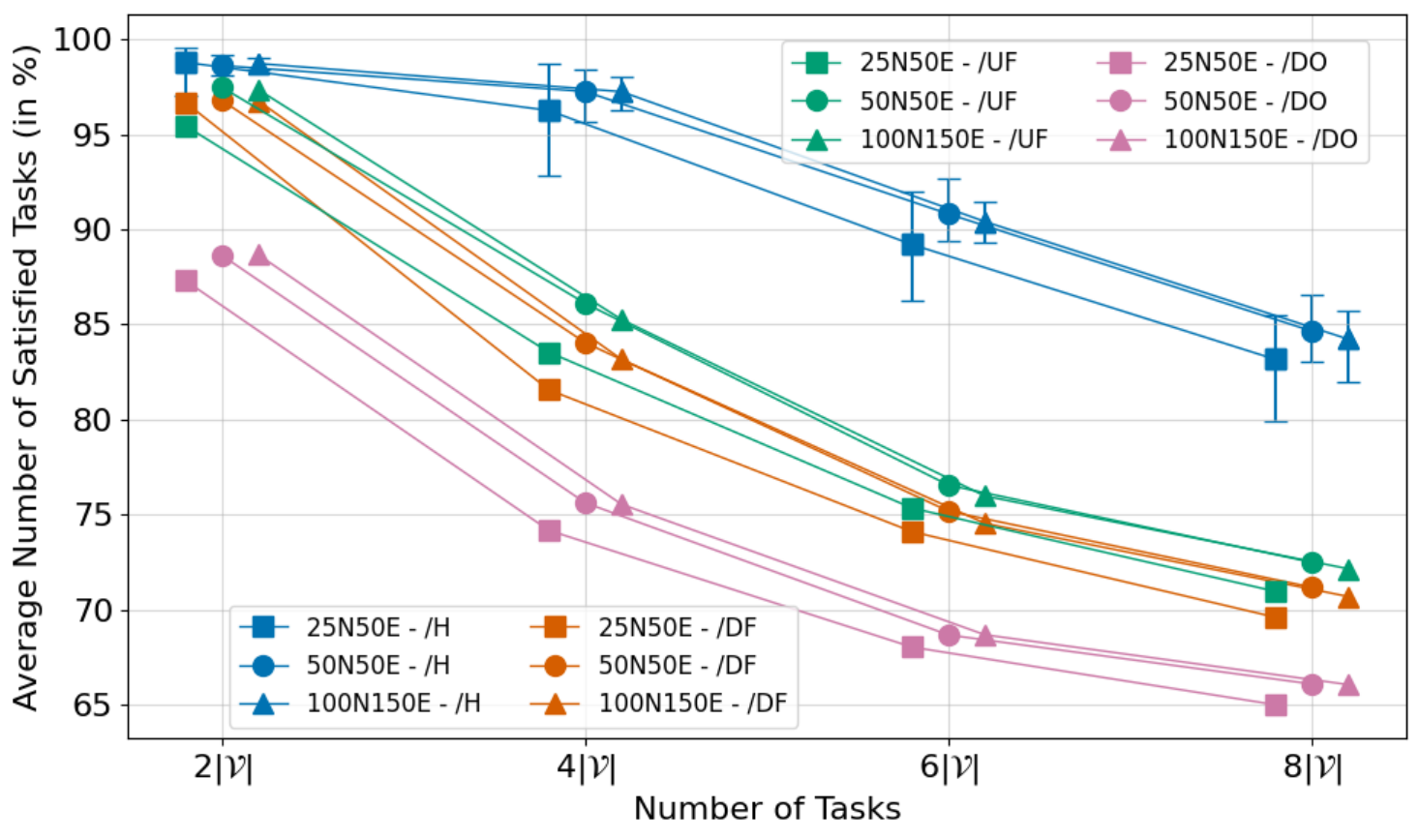}
  \caption{Number of tasks versus the percentage of satisfied tasks.}
  \label{fig_tasks}
\end{figure}
 \begin{figure}[ht]
  \centering
 \includegraphics[width=0.5\textwidth, keepaspectratio]{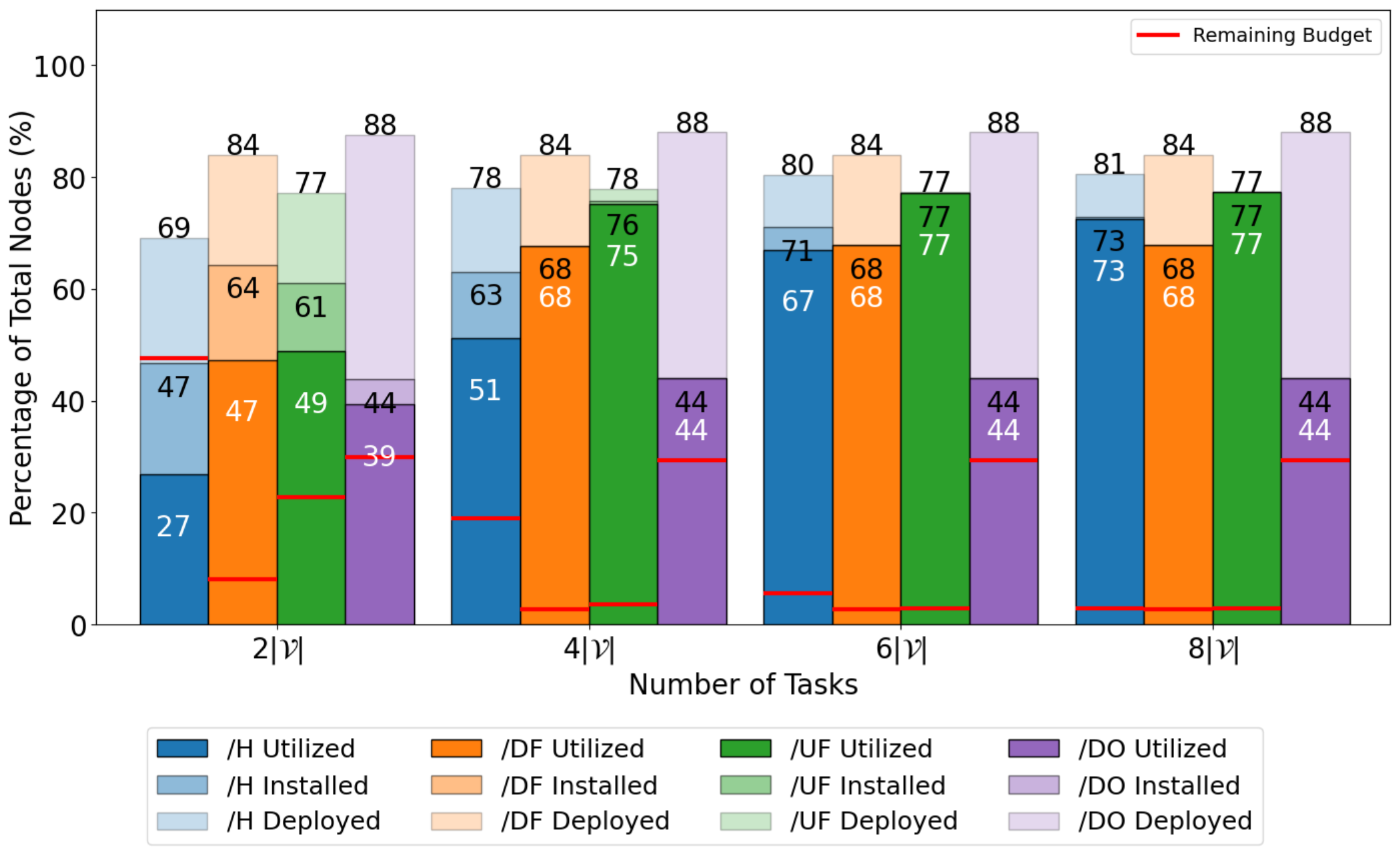}
  \caption{Number of tasks versus deployed servers, rpack installed, and server utilization.}
  \label{fig_tasks_servers2}
\end{figure}

\subsubsection{Deploy and Upgrade Cost Ratio}
\begin{figure}[ht]
  \centering
 \includegraphics[width=0.5\textwidth, keepaspectratio]{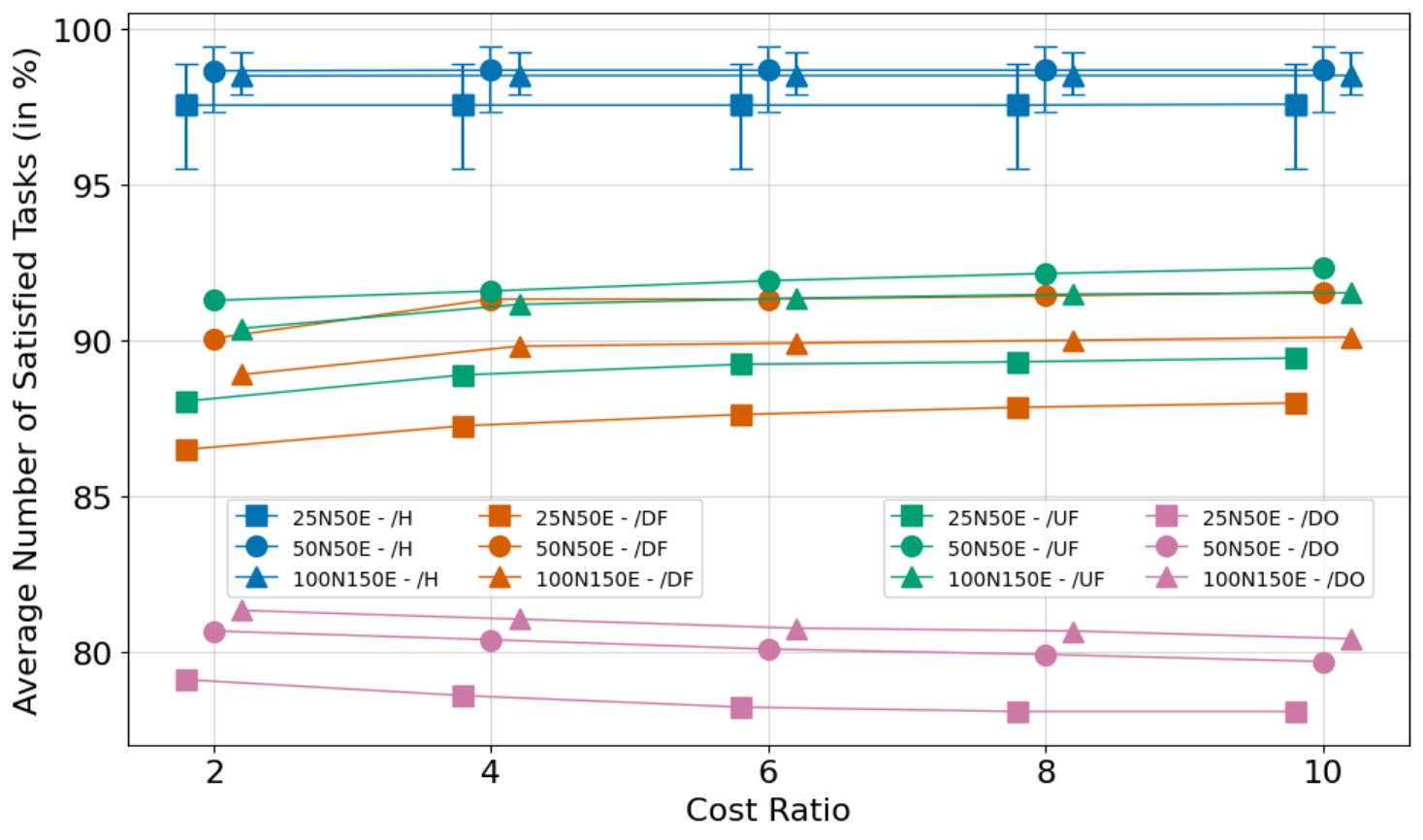}
  \caption{Cost ratio versus the percentage of satisfied tasks.}
  \label{fig_costs_small}
\end{figure}

In this section, we aim to examine the effect of a different ratio between the server deployment cost and the rpack installation cost, i.e., $\kappa^t_n:\kappa^t_u$, on the percentage of satisfied tasks $\bar{\Gamma}$ produced by /H, /HO, /DF. /UF, and /DO methods.
More specifically, Fig. \ref{fig_costs_small} shows the effect for five different $\kappa^t_n:\kappa^t_u$ ratios, i.e., 2:1, 4:1, 6:1, 8:1, and 10:1.

As shown in Fig. \ref{fig_costs_small}, for all algorithms, increasing the ratio from 2:1 to 10:1 only marginally affects $\bar{\Gamma}$, e.g., a 1\% increase in $\bar{\Gamma}$ for /H, /DF, and /UF, and a 1\% decrease for /DO.
Note that increasing the ratio from 2:1 to 10:1 also does not significantly affect the number of servers deployed $\hat{\mathcal{S}}^T$, installed rpacks $\mathcal{S}^T$, or server utilization $\mathcal{C}^T$, i.e., affecting only between 1\% and 3\% of their values. Details of the results are not presented in the paper to reduce space.

However, increasing the ratio from 2:1 to 10:1 increases the remaining budget $\hat{\Delta}B^T$, between 0.12\% and 30.74\%, for /H, /DF, and /UF, while decreasing the remaining budget, ranging from 39.09\% to 25.82\%,  for /DO.
For /H, /DF, and /UF, the remaining budget increases because as the cost ratio rises, server upgrades become relatively cheaper than server deployment. Thus, the three algorithms tend to acquire the same number of rpacks that require less budget. 
In addition, the algorithms deploy or upgrade servers only when additional tasks can be satisfied; hence, unused budget remains if no further upgrades yield additional satisfied tasks.
In contrast, /DO shows a decrease in the number of deployed servers because higher cost ratios make new deployments more expensive, thereby reducing the number of rpacks that can be installed.
 \begin{figure}[ht]
  \centering
 \includegraphics[width=0.5\textwidth, keepaspectratio]{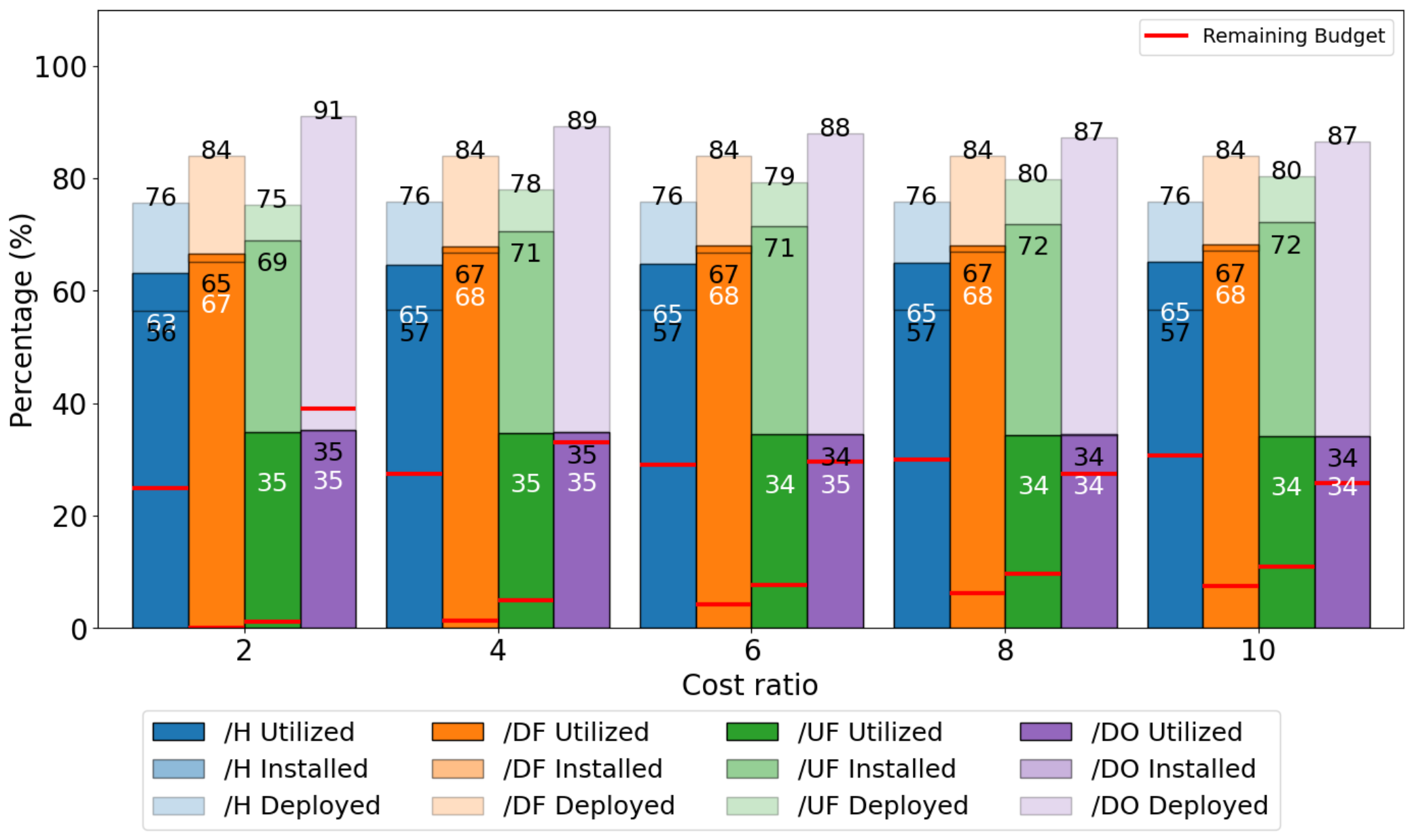}
  \caption{Cost ratio versus deployed servers, rpack installed, and server utilization.}
  \label{fig_costs_servers_v2}
\end{figure}

\subsubsection{Number of Stages}
We compare the performance of the five versions of M-ESU, i.e., /H, /DO, /DF, /UF, and /HO, when the number of stages increases from $T=1$ to $T=5$. 
The additional version /HO is identical to /H except that it does not apply task prediction at the end of stage $t=T$. This version is to evaluate the effect of incorporating \textit{task demand prediction} into the overall performance.
For a fair comparison, the results show an average of $\Bar{\Gamma}$ in five stages. As an example, for $T=3$, server deployment and upgrade are performed only in the first three stages; however,  
the task demand profiles at stage $t=4$ and $t=5$ are considered, and
the value of  $\Bar{\Gamma}$ is an average over five stages.  
We consider a task demand rate of $\sigma = 0.2$, task size growth rate of $\alpha_b$ = $0.5$, and deadline tightening rate of $\alpha_T = 0.5$. Recall that for each experiment with $T < 5$, the last stage $T$ uses the \textit{task demand prediction} discussed in Section \ref{sec_m_esu_h_overview}. Thus, we use future time horizon value $h=4$, $h=3$, $h=2$, $h=1$, and $h=0$ for the experiment with $T=1$, $T=2$, $T=3$, $T=4$, and $T=5$, respectively.  
 \begin{figure}[ht]
  \centering
 \includegraphics[width=0.5\textwidth, keepaspectratio]{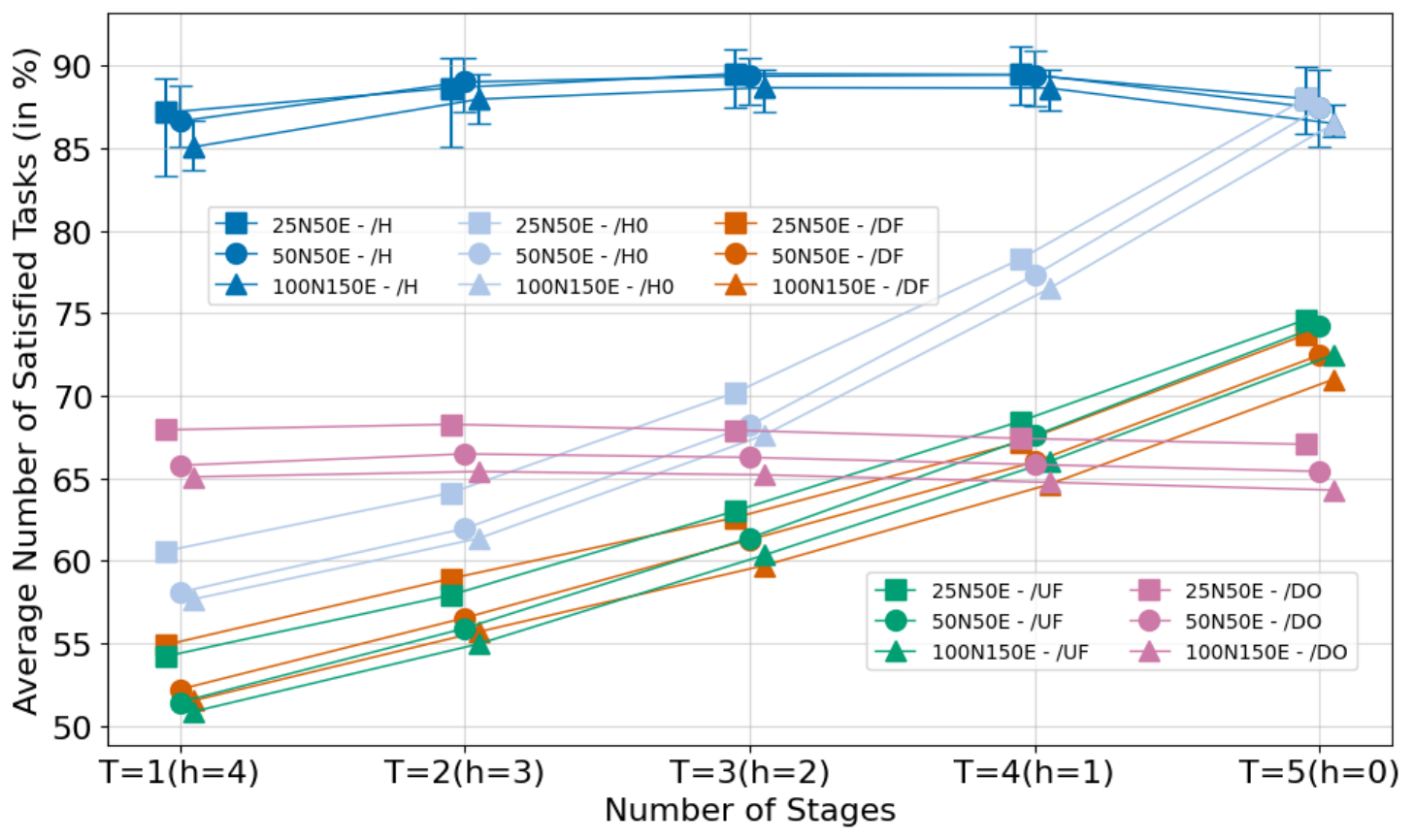}
  \caption{Number of stages versus percentage of satisfied tasks.}
  \label{fig_stages}
\end{figure}

 \begin{figure}[ht]
  \centering
 \includegraphics[width=0.5\textwidth, keepaspectratio]{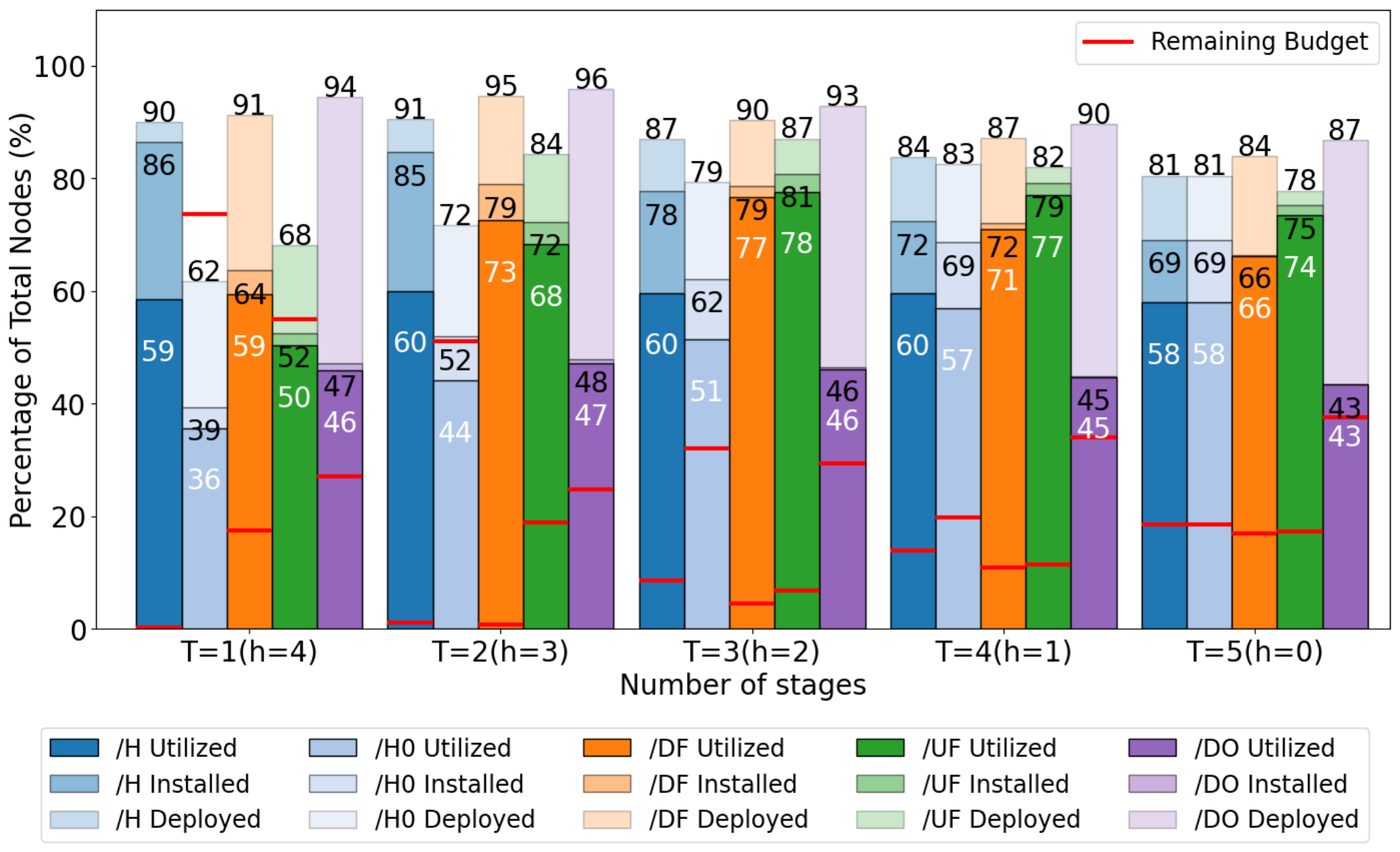}
  \caption{Number of stages versus deployed servers, installed rpacks, and server utilization.}
  \label{fig_stages_servers_util}
\end{figure}

Fig. \ref{fig_stages} shows the average percentage of satisfied tasks $\bar{\Gamma}$ for the number of stages $T$ between one and five. 
The figure shows three distinct trends for the value of $\bar{\Gamma}$ when the number of stages $T$ increases: (i) stable for /H, and (ii) increasing for /HO, /DF and /UF, and (iii) decreasing for /DO.
The figure also presents the minimum and maximum values for /H, whose standard deviation ranges from 0.62 to 1.79. These results confirm that 
/H exhibits reliable behavior across all tested values of $T$.
We use Fig.~\ref{fig_stages_servers_util} to help explain the results in Fig.~\ref {fig_stages}.

For (i) as shown in Fig.~\ref{fig_stages_servers_util}, 
the server utilization in /H increases from $\mathcal{C}^T = 59\%$ to $\mathcal{C}^T = 60\%$ as $T$ increases from one to four, and then decreases to 58\% at $T = 5$.
This result explains the increase in the average task satisfaction from $\bar{\Gamma} = 86.62\%$ to $\bar{\Gamma} = 89.43\%$ for $T = 1$ to $T = 4$, followed by a decrease to $\bar{\Gamma} = 87.41\%$ at $T = 5$, as shown in Fig.~\ref{fig_stages}.
Interestingly, Fig. \ref{fig_stages_servers_util} shows that for /H, $\hat{\mathcal{S}}^T$ and $\mathcal{C}^T$, respectively, decrease from 90\% to 81\% and 86\% to 69\% when $T$ increases from one to five. 
The results show that /H is able to manage resources more efficiently as the number of stages increases. This shows the merit of separating the deployment and upgrade of edge servers
into multiple stages because it enables operators to lower overall deployment costs without compromising performance.

For (ii), as shown in Fig.~\ref{fig_stages}, although /HO exhibits an upward trend, it still performs significantly worse than /H for $T=1$ to $T=4$.
This occurs because, at the last stage $t = T$, /HO deploys and/or upgrades  servers considering the task demands \textit{only} at that stage.
Consequently, at the next $h$ stages after $T$, when no further server deployments or upgrades are performed, the increase in the number of tasks cannot be accommodated. This result shows the advantage of using the \textit{task demand prediction} in /H.
In addition, the \textit{ prediction} used in /H helps improving budget utilization. As shown in Fig.~\ref{fig_stages_servers_util}, /HO has a larger unutilized budget compared to /H, e.g., $\hat{\Delta}B^T = 73.66\%$ versus $\hat{\Delta}B^T =0.3\%$ for $T=1$.
While /UF and /DF have the same trend as /HO, they produce lower $\bar{\Gamma}$. The result shows the effectiveness of /HO, and hence /H, in balancing between server deployment and server upgrade compared to giving priority to server upgrade in /UF and server deployment in /DF. 

As shown in  Fig.~\ref{fig_stages}, there is an upward trend for the average percentage of satisfied tasks in /DF and /UF when the value of $T$ increases. 
The reason is because spreading the available budget across a larger number of stages $T$ enables a more effective allocation of resources. 
We use Fig. \ref{fig_stages_servers_util} to explain the reason. 
For $T=1$, the algorithms consume the entire budget in one stage to deploy and  upgrade servers based on long-range, i.e., $h = 4$ stages, task-profile predictions. Ineffective server deployment and upgrade results in a low average percentage of satisfied tasks across the five stages. This observation is confirmed in Fig. \ref{fig_stages_servers_util} showing that there is a large gap between the number of servers deployed and the number of installed rpacks, i.e., /DF deploys 91\% of the servers but installs only 64\% of the total rpacks, while /UF deploys 68\% of the servers with only 52\% of rpacks installed.
In contrast, increasing the value of $T$, e.g., for $T = 5$, and thus $h = 0$, reduces the gap between deployed servers and installed rpacks, e.g., to 12 and 3 percentage points for /DF and /UF, respectively. 
The reason is because a larger value of $T$, and hence smaller $h$, allows the algorithms to make more effective sequential deployment and upgrade decisions that consider a shorter-range of task profile prediction. Thus, the algorithms are expected to yield a higher average of satisfied tasks.
The results for /DF and /UF are consistent with /H, i.e., as the number of stages $T$ increases, the network scales more efficiently, deploying fewer servers while achieving higher $C^T$.

Finally, for (iii), Fig. \ref{fig_stages} shows that /DO is significantly worse compared to /H, highlighting the benefit of including the upgrade of server capacity in /H. Notice that /DO deploys more servers but installs fewer rpacks compared to /H.
The smaller number of installed rpacks in /DO leads to near-full server utilization for all values of $T$, 
e.g., as shown in Fig. \ref{fig_stages_servers_util}, for $T = 3$ both $\hat{M}^T$ and $C^T$ are 46\%, meaning 100\% installed rpacks are being utilized, resulting in poorer performance.

\section{CONCLUSION} \label{sec_conclusion}
This paper has introduced M-ESU, a novel network upgrade problem involving the deployment and upgrading of edge servers over multiple stages. 
More precisely, M-ESU addresses a real-world challenge in which network operators aim to increase the number of tasks that meet their deadlines while operating within budget and server capacity limitations, while accommodating both growth and their increasingly strict deadline requirements over time.
Two complementary solutions were presented: a MILP formulation that guarantees optimal results for small networks, and a heuristic algorithm (M-ESU/H) that efficiently scales to large networks. 
Experimental results verify that M-ESU/H provides a near-optimal balance between performance and computational cost, outperforming baseline approaches in both task satisfaction and resource efficiency. 
The results also highlight the benefits of integrating predictive task demand, which enables operators to sustain longer service quality under constrained deployment cycles. 

Future extensions of this work may include multi-objective optimization to jointly minimize energy consumption and maximize reliability, the inclusion of dependent tasks, and maintaining task satisfaction rate in evolving MEC networks and their task demands.

\section{ACKNOWLEDGMENT}
Appreciation is extended to Sebelas Maret University for providing exceptional support in sponsoring this research.

\section*{REFERENCES}
\renewcommand*{\bibfont}{\footnotesize}
\printbibliography[heading=none]

@article{yang2019cloudlet,
  title={Cloudlet placement and task allocation in mobile edge computing},
  author={Yang, Song and Li, Fan and Shen, Meng and Chen, Xu and Fu, Xiaoming and Wang, Yu},
  journal={IEEE Internet Things J.},
  volume={6},
  number={3},
  pages={5853--5863},
  month=mar,
  year={2019}
}

@article{ren2021demand, 
title={A Demand-Driven Incremental Deployment Strategy for Edge Computing in {IoT} Network},
  author={Ren, Wei and Sun, Yan and Luo, Hong and Guizani, Mohsen},
  journal={IEEE Trans. Netw. and Eng.},
  volume={9},
  number={2},
  pages={416--430},
  month=oct,
  year={2021}
}

@article{li2021placement,
  title={Placement of edge server based on task overhead in mobile edge computing environment},
  author={Li, Bo and Hou, Peng and Wu, Hao and Qian, Rongrong and Ding, Hongwei},
  journal={Trans. Emerg. Telecommun. Technol},
  volume={32},
  number={9},
  pages={e4196},
  month=sep,
  year={2021}
}

@article{wang2019edge,
  title={Edge server placement in mobile edge computing},
  author={Wang, Shangguang and Zhao, Yali and Xu, Jinlinag and Yuan, Jie and Hsu, Ching-Hsien},
  journal={J. Parallel Distrib. Comput.},
  volume={127},
  pages={160--168},
  month=may,
  year={2019}
}

@article{wang2017controller,
  title={The controller placement problem in software defined networking: A survey},
  author={Wang, Guodong and Zhao, Yanxiao and Huang, Jun and Wang, Wei},
  journal={IEEE Netw.},
  volume={31},
  number={5},
  pages={21--27},
  month=sep,
  year={2017}
}

@article{lahderanta2021edge,
  title={Edge computing server placement with capacitated location allocation},
  author={L{\"a}hderanta, Tero and Lepp{\"a}nen, Teemu and Ruha, Leena and Lov{\'e}n, Lauri and Harjula, Erkki and Ylianttila, Mika and Riekki, Jukka and Sillanp{\"a}{\"a}, Mikko J},
  journal={J. Parallel Distrib. Comput.},
  volume={153},
  pages={130--149},
  month=jul,
  year={2021}
}

@article{xiang2021dataset,
  title={A dataset for mobile edge computing network topologies},
  author={Xiang, Bin and Elias, Jocelyne and Martignon, Fabio and Di Nitto, Elisabetta},
  journal={Data in Brief},
  volume={39},
  pages={107557},
  month=dec,
  year={2021}
}

@inproceedings{buzachis2018innovative,
  title={An innovative MapReduce-based approach of Dijkstra’s algorithm for {SDN} routing in hybrid cloud, edge and {IoT} scenarios},
  author={Buzachis, Alina and Galletta, Antonino and Celesti, Antonio and Villari, Massimo},
  booktitle={Proc. 7th Eur. Conf. Service-Oriented Cloud Comput. (ESOCC)},
  address = {Como, Italy},
  pages={185--198},
  month=sep,
  year={2018}
}

@article{tiwari2024knapsack,
  title={A Knapsack-based Metaheuristic for Edge Server Placement in 5G networks with heterogeneous edge capacities},
  author={Tiwari, Vaibhav and Pandey, Chandrasen and Dahal, Abisek and Roy, Diptendu Sinha and Fiore, Ugo},
  journal={Future Gener. Comput. Syst.},
  volume={153},
  pages={222--233},
  month=apr,
  year={2024}
}

@article{wang2022optimal,
  title={An optimal edge server placement approach for cost reduction and load balancing in intelligent manufacturing},
  author={Wang, Zhongmin and Zhang, Weiye and Jin, Xiaomin and Huang, Yihua and Lu, Chen},
  journal={J. Supercomput.},
  volume={78},
  number={3},
  pages={4032--4056},
  month=feb,
  year={2022}
}

@inproceedings{badnava2023energy,
  title={Energy-efficient deadline-aware edge computing: Bandit learning with partial observations in multi-channel systems},
  author={Badnava, Babak and Roach, Keenan and Cheung, Kenny and Hashemi, Morteza and Shroff, Ness B},
  booktitle={IEEE Globecom},
  address={Kuala Lumpur, Malaysia},
  pages={3081--3086},
  month=dec,
  year={2023}
}

@article{paik1995network,
  title={Network upgrading problems},
  author={Paik, Doowon and Sahni, Sartaj},
  journal={Networks},
  volume={26},
  number={1},
  pages={45--58},
  month=aug,
  year={1995}
}

@article{song2022delay,
  title={Delay-sensitive tasks offloading in multi-access edge computing},
  author={Song, Shudian and Ma, Shuyue and Yang, Lingyu and Zhao, Jingmei and Yang, Feng and Zhai, Linbo},
  journal={Expert Syst. Appl.},
  volume={198},
  pages={116730},
  month=jul,
  year={2022}
}

@article{wihidayat2025multi,
  title={Delay-aware multi-stage edge server placement and task offloading with budget constraint},
  author={Wihidayat, Endar Suprih and   Soh, Sieteng and Chin, Kwan-Wu and Pham, Duc-Son},
  journal ={Comput. Netw.},
  volume = {269},
  pages = {111413},
  month = sep,
  year = {2025}
}

@article{ranaweera2020novel,
  title={Novel MEC based approaches for smart hospitals to combat COVID-19 pandemic},
  author={Ranaweera, Pasika and Liyanage, Madhusanka and Jurcut, Anca Delia},
  journal={IEEE Consum. Electron. Mag.},
  volume={10},
  number={2},
  pages={80--91},
  month=oct,
  year={2020}
}

@inproceedings{loven2020scaling,
  title={Scaling up an edge server deployment},
  author={Lov{\'e}n, Lauri and L{\"a}hderanta, Tero and Ruha, Leena and Lepp{\"a}nen, Teemu and Peltonen, Ella and Riekki, Jukka and Sillanp{\"a}{\"a}, Mikko J},
  booktitle={Proc. IEEE Int. Conf. Pervasive Comput. Commun. Workshops (PerCom Workshops)},
  address={Austin, TX, USA},
  pages={1--7},
  month=mar,
  year={2020}
}

@article{cruz2022edge,
  title={On the edge of the deployment: A survey on multi-access edge computing},
  author={Cruz, Pedro and Achir, Nadjib and Viana, Aline Carneiro},
  journal={ACM Comput. Surv.},
  volume={55},
  number={5},
  pages={1--34},
  month=dec,
  year={2022}
}

@inproceedings{nigade2021better,
  title={Better never than late: Timely edge video analytics over the air},
  author={Nigade, Vinod and Winder, Ramon and Bal, Henri and Wang, Lin},
  booktitle={Proc. ACM Conf. Embedded Netw. Sens. Syst. (SenSys)},
  pages={426--432},  
  numpages = {7},
  address = {Coimbra, Portugal},
  month = nov,
  year={2021}
}

@article{xu2025budget,
  title={Budget-Constrained Edge Server Expansion Deployment via Genetic Algorithm and Particle Swarm Optimization},
  author={Xu, Qinglong and Gao, Zhenguo and Jiang, Yang and Gan, Qiren and Zhao, Yunlong and Wu, Hsiao-Chun},
  journal={IEEE Internet Things J.}, 
  volume={12},
  number={21},
  pages={45501-45516},
  month = nov,
  year={2025}
}

@article{niu2025esd,
  title={{A--ESD}: Auxiliary Edge-Server Deployment for Load Balancing in Mobile Edge Computing},
  author={Niu, Sen and Zhang, Xuewei and Wang, Simin and Liao, Kaili and Zhang, Bofeng and Zou, Guobing},
  journal={Math.},
  volume={13},
  number={19},
  pages={3087},
  month=sep,
  year={2025}
}

@article{mao2024green,
  author={Mao, Yuyi and Yu, Xianghao and Huang, Kaibin and Angela Zhang, Ying-Jun and Zhang, Jun},
  title={Green Edge {AI}: A Contemporary Survey}, 
  journal={Proc. IEEE}, 
  volume={112},
  number={7},
  pages={880-911},
  month = jul,
  year={2024}
}

@article{zhang2024survey,
  title={A Survey of Computation Offloading With Task Types}, 
  author={Zhang, Siqi and Yi, Na and Ma, Yi},
  journal={IEEE Trans. Intell. Transp. Syst.}, 
  volume={25},
  number={8},
  pages={8313-8333},
  month =aug,
  year={2024}
}

@article{dong2024task,
  title={Task offloading strategies for mobile edge computing: A survey},
  author={Dong, Shi and Tang, Junxiao and Abbas, Khushnood and Hou, Ruizhe and Kamruzzaman, Joarder and Rutkowski, Leszek and Buyya, Rajkumar},
  journal={Comput. Netw.},
  volume={254},
  pages={110791},
  month=dec,
  year={2024}
}

@article{li2018computation,
  title={Computation offloading algorithm for arbitrarily divisible applications in mobile edge computing environments: An OCR case},
  author={Li, Bo and He, Min and Wu, Wei and Sangaiah, Arun Kumar and Jeon, Gwanggil},
  journal={Sustainability},
  volume={10},
  number={5},
  pages={1611},
  month=may,
  year={2018},
}

@article{wu2025fairness,
  title={Fairness-Aware Budgeted Edge Server Placement for Connected Autonomous Vehicles}, 
  author={Wu, Jintao and Xu, Xiaolong and Cui, Guangming and Zhang, Yiwen and Qi, Lianyong and Dou, Wanchun and Cai, Zhipeng},
  journal={IEEE Trans. Mob. Comput.}, 
  volume={24},
  number={6},
  pages={4762-4776},
  month = jun,
  year={2025}
}

@article{wang2020joint,
  title={Joint offloading and charge cost minimization in mobile edge computing},
  author={Wang, Kehao and Hu, Zhixin and Ai, Qingsong and Zhong, Yi and Yu, Jihong and Zhou, Pan and Chen, Lin and Shin, Hyundong},
  journal={IEEE Open J. Commun. Soc.},
  volume={1},
  pages={205--216},
  month=feb,
  year={2020}
}

@article{ogawa2025transfer,
  title={Transfer Probability-Based Job Reallocation Method for Heterogeneous Edge Clouds},
  author={Ogawa, Kohei and Miyata, Sumiko and Kanai, Kenji},
  journal={IEEE Open J. Commun. Soc.},
  volume={6},
  number={},
  pages={4549-4562},
  month=may,
  year={2025}
}

@article{shibata2025edge, 
  title={Edge Server Placement and Task Allocation for Maximum Delay Reduction},
  author={Shibata, Koki and Miyata, Sumiko},
  journal={IEEE Open J. Commun. Soc.},
  volume={6},
  pages={6207-6217},
  month=jul,
  year={2025}  
}

@article{zeng2022parallel,
  title={Parallel Processing at the Edge in Dense Wireless Networks},
  author={Zeng, Ming and Fodor, Viktoria},
  journal={IEEE Open J. Commun. Soc.},  
  volume={3},
  pages={1-14},
  month=jan,
  year={2022}
}

@article{asghari2024server,
  title={Server placement in mobile cloud computing: A comprehensive survey for edge computing, fog computing and cloudlet},
  author={Asghari, Ali and Sohrabi, Mohammad Karim},
  journal={Computer Science Review},
  volume={51},
  pages={100616},
  month=feb,
  year={2024}
}

@article{xu2024unleashing,
  title={Unleashing the power of edge-cloud generative AI in mobile networks: A survey of AIGC services},
  author={Xu, Minrui and Du, Hongyang and Niyato, Dusit and Kang, Jiawen and Xiong, Zehui and Mao, Shiwen and Han, Zhu and Jamalipour, Abbas and Kim, Dong In and Shen, Xuemin and others},
  journal={IEEE Commun. Surv. Tut.},
  volume={26},
  number={2},
  pages={1127--1170},
  month=jan,
  year={2024}
}

\begin{IEEEbiography}[{\includegraphics[width=1in,height=1.25in,clip,keepaspectratio]{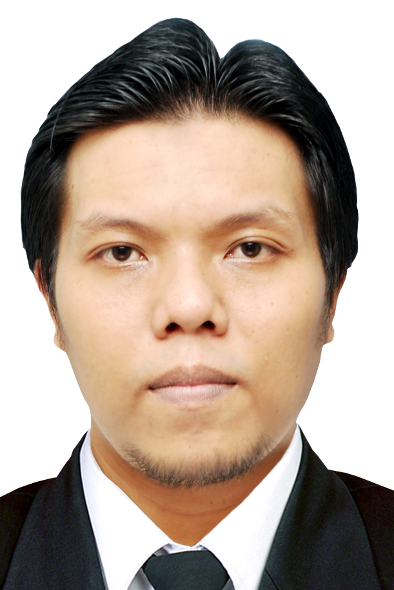}}]{\uppercase{Endar S. Wihidayat}}~earned his Bachelor's degree in Electrical Engineering from Diponegoro University, Indonesia, in 2003. He later completed a Master of Engineering at Gadjah Mada University, Indonesia, in 2012. From 2004 to 2013, he worked as a Telecommunication Engineer at Siemens. Currently, he is a lecturer at Sebelas Maret University, Indonesia, and is also pursuing a PhD at Curtin University. His main focus of research is network planning and optimization.
\end{IEEEbiography}
\begin{IEEEbiography}[{\includegraphics[width=1in,height=1.25in,clip,keepaspectratio]{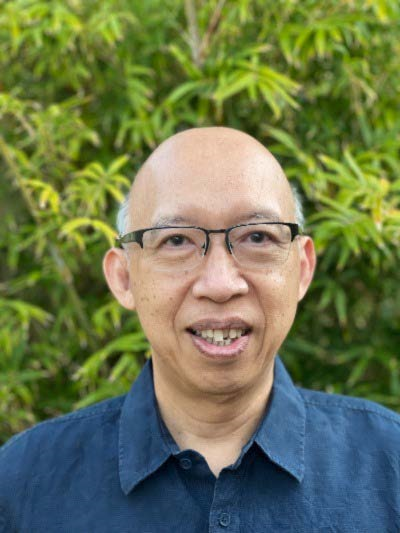}}]{\uppercase{Sieteng Soh}}~(Member, IEEE) received the BS degree in electrical engineering from the University of Wisconsin-Madison, in 1987, and the MS and Ph.D. degrees in electrical engineering from Louisiana State University, Baton Rouge, in 1989 and 1993, respectively. From 1993 to 2000, he was with the Tarumanagara University, Indonesia. He is currently an Associate Professor with the School of Electrical Engineering, Computing and Mathematical Sciences, Curtin University, Perth, Australia. He has published over 150 international journals and conference papers. His current research interests include algorithm design, network optimization, and network reliability.
\end{IEEEbiography}
\begin{IEEEbiography}[{\includegraphics[width=1in,height=1.25in,clip,keepaspectratio]{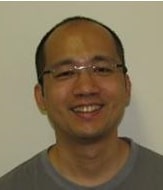}}]{\uppercase{Kwan-Wu Chin}}~received the Bachelor of Science degree (Hons.) and the Ph.D. degree with commendation from Curtin University, Australia, in 1997 and 2000, respectively. He was a Senior Research Engineer at Motorola, from 2000 to 2003. In 2004, he joined the University of Wollongong as a Senior Lecturer. He is currently an Associate Professor. He holds four United States of America (USA) patents and has published over 200 journals and conference papers. His research interests include medium access control protocols for wireless networks and resource allocation algorithms/policies for communications networks.
\end{IEEEbiography}
\begin{IEEEbiography}[{\includegraphics[width=1in,height=1.25in,clip,keepaspectratio]{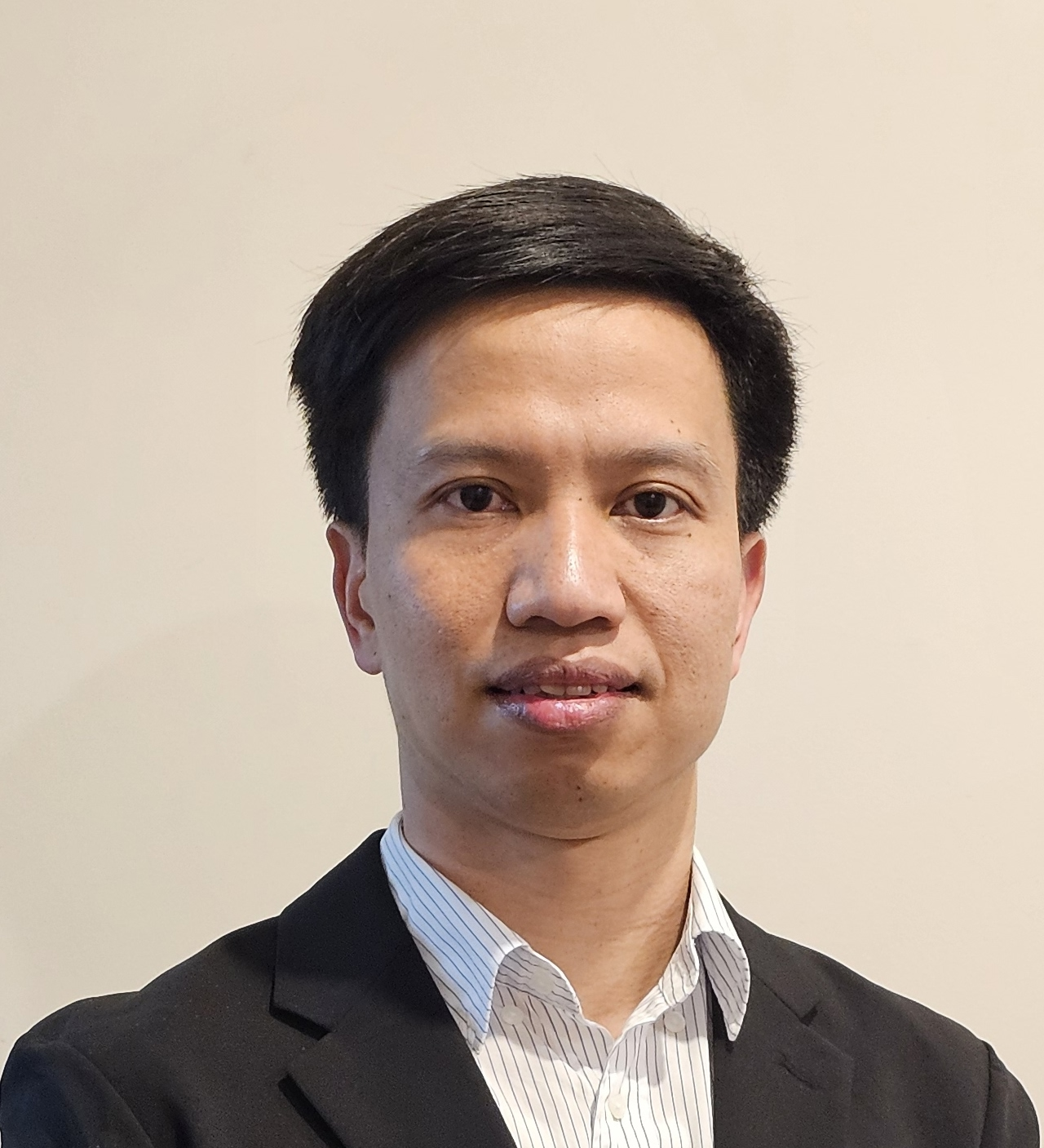}}]{\uppercase{Duc-son Pham}}~(Senior Member, IEEE) received the Ph.D. degree from the Curtin University of Technology, in 2005. He is currently an Associate Professor with the Discipline of Computing, Curtin University, and Perth, WA, Australia. His current research interests include sparse learning theory, large-scale data mining, convex optimization, and advanced deep learning with applications to computer vision and image processing. He was a recipient of the Young Author Best Paper Award 2010 for a publication in \emph{IEEE Transactions on Signal Processing.}
\end{IEEEbiography}
\EOD
\end{document}